\documentclass[11pt, a4paper]{article}

\usepackage{times}
\usepackage{wrapfig}
\usepackage{sectsty}
\sectionfont{\large}
\subsectionfont{\normalsize}
\subsubsectionfont{\small}
\usepackage{listings}
\usepackage{caption}
\usepackage{subcaption}

\usepackage{amsfonts}
\usepackage{microtype}
\usepackage{graphicx}

\usepackage{booktabs}

\usepackage{multicol}
\usepackage{algorithm}
\usepackage{algpseudocode}
\usepackage{enumitem}

\usepackage{setspace}
\usepackage{natbib}

\usepackage{amsmath}
\usepackage{amssymb}
\usepackage{mathtools}
\usepackage{amsthm}

\usepackage[colorlinks,
            linkcolor=red,
            anchorcolor=blue,
            citecolor=blue
            ]{hyperref}
\usepackage[capitalize,noabbrev]{cleveref}

\DeclareFontFamily{U}{mathb}{}
\DeclareFontShape{U}{mathb}{m}{n}{
      <-6> mathb5
      <6-7> mathb6
      <7-8> mathb7
      <8-9> mathb8
      <9-10> mathb9
      <10-12> mathb10
      <12-> mathb12
}{}
\DeclareSymbolFont{mathb}{U}{mathb}{m}{n}

\DeclareMathSymbol{\drsh}{3}{mathb}{"EB}

\theoremstyle{plain}
\newtheorem{theorem}{Theorem}[section]
\newtheorem{proposition}[theorem]{Proposition}
\newtheorem{lemma}[theorem]{Lemma}

\newtheorem{example}[theorem]{Example}

\theoremstyle{definition}
\newtheorem{definition}[theorem]{Definition}
\newtheorem{assumption}[theorem]{Assumption}

\usepackage{smile}
\usepackage{makecell}

\usepackage{amsmath,amsfonts,bm}

\def\1{\bm{1}}

\def\vzero{{\bm{0}}}
\def\vone{{\bm{1}}}

\def\vbeta{{\bm{\beta}}}

\def\va{{\bm{a}}}
\def\vb{{\bm{b}}}
\def\vc{{\bm{c}}}
\def\vd{{\bm{d}}}
\def\ve{{\bm{e}}}
\def\vf{{\bm{f}}}

\def\vp{{\bm{p}}}
\def\vq{{\bm{q}}}

\def\vs{{\bm{s}}}

\def\vu{{\bm{u}}}
\def\vv{{\bm{v}}}
\def\vw{{\bm{w}}}
\def\vx{{\bm{x}}}
\def\vy{{\bm{y}}}
\def\vz{{\bm{z}}}

\def\evbeta{{\beta}}

\def\eva{{a}}

\def\evc{{c}}
\def\evd{{d}}

\def\evf{{f}}

\def\evp{{p}}
\def\evq{{q}}

\def\evu{{u}}
\def\evv{{v}}

\def\evx{{x}}
\def\evy{{y}}
\def\evz{{z}}

\def\mA{{\bm{A}}}
\def\mB{{\bm{B}}}

\def\mF{{\bm{F}}}

\def\mLambda{{\bm{\Lambda}}}
\def\mSigma{{\bm{\Sigma}}}

\DeclareMathAlphabet{\mathsfit}{\encodingdefault}{\sfdefault}{m}{sl}
\SetMathAlphabet{\mathsfit}{bold}{\encodingdefault}{\sfdefault}{bx}{n}

\def\gA{{\mathcal{A}}}

\def\gG{{\mathcal{G}}}

\def\gK{{\mathcal{K}}}

\def\gM{{\mathcal{M}}}
\def\gN{{\mathcal{N}}}
\def\gO{{\mathcal{O}}}
\def\gP{{\mathcal{P}}}
\def\gQ{{\mathcal{Q}}}

\def\gW{{\mathcal{W}}}
\def\gX{{\mathcal{X}}}
\def\gY{{\mathcal{Y}}}
\def\gZ{{\mathcal{Z}}}

\def\sR{{\mathbb{R}}}

\def\emLambda{{\Lambda}}

\def\emSigma{{\Sigma}}

\DeclareMathOperator*{\argmax}{arg\,max}
\DeclareMathOperator*{\argmin}{arg\,min}

\newcommand{\pluseq}{\mathrel{{+}{=}}}

\newcommand*{\T}{\mathsf{T}}
\DeclareMathOperator\supp{supp}

\newcommand{\interior}[1]{%
  {\kern0pt#1}^{\mathrm{o}}%
}

\newcommand\extrafootertext[1]{%
    \bgroup
    \renewcommand\thefootnote{\fnsymbol{footnote}}%
    \renewcommand\thempfootnote{\fnsymbol{mpfootnote}}%
    \footnotetext[0]{#1}%
    \egroup
}
\usepackage[textsize=tiny]{todonotes}
\usepackage{fullpage}

\usepackage{float}
\usepackage{titlesec}

\titleformat{\section}
  {\normalfont\large\bfseries}{\thesection}{1em}{}
\titleformat{\subsection}
  {\normalfont\normalsize\bfseries}{\thesubsection}{1em}{}
\titleformat{\subsubsection}
  {\normalfont\normalsize\bfseries}{\thesubsubsection}{1em}{}
\titlespacing*{\section}{0pt}{5pt plus 1pt minus 1pt}{2pt plus 1pt minus 1pt}

\titlespacing*{\subsection}{0pt}{4pt plus 1pt minus 1pt}{2pt plus 1pt minus 1pt}

\titlespacing*{\subsubsection}{0pt}{4pt plus 1pt minus 1pt}{1pt plus 1pt minus 1pt}

\labelwidth\leftmargini\advance\labelwidth-\labelsep \labelsep 5pt

\def\@listi{\leftmargin\leftmargini}
\def\@listii{\leftmargin\leftmarginii
   \labelwidth\leftmarginii\advance\labelwidth-\labelsep
   \topsep 2pt plus 1pt minus 0.5pt
   \parsep 1pt plus 0.5pt minus 0.5pt
   \itemsep \parsep}
\def\@listiii{\leftmargin\leftmarginiii
    \labelwidth\leftmarginiii\advance\labelwidth-\labelsep
    \topsep 1pt plus 0.5pt minus 0.5pt
    \parsep \z@ \partopsep 0.5pt plus 0pt minus 0.5pt
    \itemsep \topsep}
\def\@listiv{\leftmargin\leftmarginiv
     \labelwidth\leftmarginiv\advance\labelwidth-\labelsep}
\def\@listv{\leftmargin\leftmarginv
     \labelwidth\leftmarginv\advance\labelwidth-\labelsep}
\def\@listvi{\leftmargin\leftmarginvi
     \labelwidth\leftmarginvi\advance\labelwidth-\labelsep}
     
\title{\textbf{\Large{
Differentiable Bilevel Programming for Stackelberg Congestion Games}}}
\author
{\normalsize
Jiayang Li$^{1}$
\qquad Jing Yu$^{1}$
\qquad Qianni Wang$^{1}$
\qquad Boyi Liu$^{2}$\\
\normalsize
\qquad Zhaoran Wang$^{2}$
\qquad Yu (Marco) Nie$^{1,*}$
}

\date{}

\begin{document}

\extrafootertext{$^1$Department of Civil and Environmental Engineering, Northwestern University, Evanston, IL 60208, USA. $^2$Department of Industrial Engineering and Management Science, Northwestern University, Evanston, IL 60208, USA. $^*$Corresponding author; email: \texttt{y-nie@northwestern.edu}.
}

\maketitle

\begin{abstract}

In a Stackelberg congestion game (SCG), a leader aims to maximize their own gain by anticipating and manipulating the equilibrium state at which the followers settle by playing a congestion game. Often formulated as bilevel programs, large-scale SCGs are well known for their intractability and complexity. Here, we attempt to tackle this computational challenge by marrying traditional methodologies with the latest differentiable programming techniques in machine learning. The core idea centers on replacing the lower-level equilibrium problem with a smooth evolution trajectory defined by the imitative logit dynamic (ILD), which we prove converges to the equilibrium of the congestion game under mild conditions.  Building upon this theoretical foundation, we propose two new local search algorithms for SCGs. 
The first is a gradient descent algorithm that obtains the derivatives by unrolling  ILD  via differentiable programming. Thanks to the smoothness of ILD, the algorithm promises both efficiency and scalability. The second algorithm adds a heuristic twist by cutting short the followers’ evolution trajectory. Behaviorally, this means that, instead of anticipating the followers' best response at equilibrium, the leader seeks to approximate that response by only looking ahead a limited number of steps.  Our numerical experiments are carried out over various instances of classic SCG applications, ranging from toy benchmarks to large-scale real-world examples.  The results show the proposed algorithms are reliable and scalable local solvers that deliver high-quality solutions with greater regularity and significantly less computational effort compared to the many incumbents included in our study.
\end{abstract}

\section{Introduction}
\label{sec:introduction}

In a Stackelberg game, a leader aims to maximize their own gain by manipulating the self-interested followers of the game \citep{von1952theory, sherali1983stackelberg}. Such a game has a natural \textit{bilevel} hierarchy. At the upper level, the leader's decision is contingent upon the response of the followers; at the lower level, the response of the followers is regulated by that very decision. Many real-world applications arising from transportation systems can be framed as a Stackelberg game in which a leader aims to improve the performance of a transportation system used by travelers (followers).
These applications range from network pricing \citep[e.g.,][]{dafermos1973toll,smith1979marginal,brotcorne2001bilevel, verhoef2002second,friesz2004dynamic,roch2005approximation,simoni2019congestion,lazar2019optimal,li2021traffic,delle2021pricing} and network design \citep[e.g.,][]{leblanc1975algorithm,smith1979traffic,marcotte1992efficient, yang1998models,meng2001equivalent, zhang2009active,li2012global,chen2016optimal,chen2017optimal} to traffic control \citep[e.g.,][]{gartner1985demand,improta1987mathematical,cantarella1987methods, smith1993traffic,yang1994traffic,peeta1995multiple, yang2007stackelberg,levin2016cell,lu2019trajectory}. 
In these applications, the followers' best response is usually characterized as a Wardrop equilibrium (WE) of a nonatomic congestion game, which dictates no traveler can reduce their travel time by adjusting travel choices unilaterally  \citep{wardrop1952road}. Hence, a \emph{Stackelberg congestion game} (SCG), the focus of the present study, is equivalent to a bilevel program constrained by a Wardrop equilibrium, a special type of \emph{Mathematical Program with Equilibrium Constraints} (MPEC) \citep{luo1996mathematical}.

\subsection{Challenges}

{Bilevel programs are notorious for their complexity. \citet{jeroslow1985polynomial} showed that bilevel linear programs (BLP) are NP-hard.  This result was later refined by \citet{vicente1994descent}, who proved that just verifying whether a given solution to a BLP is a local minimum is NP-hard. As the WE problem can be formulated as a general convex program under standard assumptions, SCGs would be equally hard, if not harder, compared to a BLP.  Indeed, previous studies have confirmed the strong NP-hardness of many variants of SCGs \citep[see, e.g.,][]{roch2005approximation,10.1007/978-3-540-92185-1_35, bhaskar2013, gairing2017complexity}. Given this inherent complexity, finding an ``exact" solution for SCGs is computationally prohibitive except for small instances or problems with special structures \citep[see, e.g.,][for some SCGs that can be solved in polynomial time]{marcotte1986network, gairing2017complexity}.

To solve an SCG, a standard approach is to convert the lower-level WE problem into {explicit} constraints. For example, the so-called KKT method \citep{bard1982explicit} replaces the lower-level problem with its equivalent KKT conditions, which can be transformed into linear or nonlinear constraints by introducing auxiliary integer variables. Another popular method \citep[e.g.,][]{marcotte1983network}  reformulates the lower level as a variational inequality problem (VIP). This turns the SCG into a semi-infinite program (SIP) since a VIP may be viewed as containing an {infinite} number of inequalities.  Under appropriate conditions, the resulting SIP's feasible set can be approximated by a {finite} number of ``cuts," which may be iteratively generated \citep{lawphongpanich2004mpec}. Thus,  one way or the other, the basic idea is to recast the SCG as a single-level program,  thereby opening the door to a broader array of solution techniques. However,  these single-level equivalent problems often face their own challenges. For example, the KKT method has to deal with many integer variables. Likewise, the performance of the cutting plane method is often severely compromised by the fact that the cuts needed to approximate the VIP are not only numerous but also highly nonlinear. 

For SCGs of considerable sizes,   a good stationary point is often the best that one could hope to achieve.
Yet, even this modest goal can be difficult to attain. Standard gradient-based algorithms, for example, require repeatedly computing the derivative of the leader's objective --- which depends on the best response of the followers, or the WE of the congestion game ---  with respect to the decisions. In literature, this task is often done by performing a sensitivity analysis of the lower-level WE problem \citep{friesz1990sensitivity, yang1998models}.  Despite its popularity, the sensitivity analysis-based (SAB) method is not sufficiently effective for large SCGs.  
The primary reason is that such analysis is typically carried out {implicitly} on the equilibrium conditions, which requires storing and inverting matrices \citep{tobin1988sensitivity, dafermos1988sensitivity, yang2007sensitivity}.  These matrices can have hundreds of thousands of rows and/or columns, if not more, in SCGs that arise from real-world applications.  }

\subsection{Motivation}
Recently, the machine learning (ML) community has discovered several new applications of bilevel programming.  A common feature in these bilevel programs is that their lower level is a deep learning problem \citep{goodfellow2016deep}. In other words, they can be interpreted as a Stackelberg game in which the follower trains a deep neural network (DNN) {according to the instruction of a leader}. { 
This training process, encoded as a numeric computer program, can be differentiated using \textit{differentiable programming}, a computational paradigm that finds the gradient of an algorithmic structure based on automatic differentiation (AD).} AD exploits the fact that a computer program, however complicated it may be, can be decomposed into a sequence of elementary arithmetic operations and functions \citep{baydin2018automatic}. By applying the chain rule recursively to these operations and functions, the derivative of the program's outputs with respect to any inputs can be computed automatically and efficiently
\citep{griewank1989automatic}. The concept of differentiable programming is general and flexible. For example, DNN itself is a differentiable program. In fact, the ability to quickly ``unroll" --- i.e., calculate the gradient for --- extremely complex DNN using AD was a key to the success of modern ML enterprises \citep{rumelhart1986learning, lecun1996effiicient}. However, one may also view an algorithm that trains a DNN as a differentiable program and then unrolls it with AD.

{Having witnessed the power of AD to tackle ML-inspired bilevel programs, we set out here to investigate whether that power can be harnessed for solving SCGs. Our hypothesis is that it may be much more efficient to obtain a descent direction for SCGs with an AD-based method --- \emph{by representing the solution to its lower-level WE problem as the output of a differentiable program (DiP)} --- than conventional sensitivity analysis-based (SAB) methods. Intuitively, this DiP may be coded according to an algorithm devised to solve the WE problem. 
However, not every algorithm can fulfill this function, as DiP requires smoothness. For example, the algorithms relying on the shortest path routine \citep[e.g., the algorithm by][]{frank1956algorithm}  evidently violate this condition. \citet{li2020end} proposed to cast the Euclidean projection algorithm \citep{bertsekas1982projection} as a DiP.  Yet, AD is not directly applicable in this case either, because the resulting DiP consists of a series of quadratic programs for which no analytic solutions are available. \citet{li2020end}  bypassed this obstacle by employing a recent DiP toolbox \citep{agrawal2019differentiable} that can handle convex programs. However, since the toolbox relies on implicit differentiation to compute the derivatives of convex programs, it encounters the same scalability issue that has haunted the SAB methods.  Therefore, while \citet{li2020end} offered a conceptual innovation, their method was not a breakthrough in terms of computation. 
}

\subsection{Main contributions}

{
In view of the computational challenges that have stalled the real-world application of SCGs, we aim to boost the scalability of local search algorithms for SCGs by leveraging the power of AD.

To this end, we first develop a DiP formulation of WE that can be efficiently unrolled with AD. 
Specifically, we identify a class of zeroth-order algorithms originating from evolutionary game theory \citep{weibull1997evolutionary, sandholm2015population} as promising candidates. In transportation literature, such algorithms are typically used to explain whether and how WE can be reached by travelers with realistic behaviors.  
In particular, we show that the imitative logit dynamical (ILD) model \citep{bjornerstedt1994nash} is a good fit because it  (i) guarantees convergence toward WE under mild conditions, and (ii) scales well in  AD-based unrolling.  

We propose two new local search algorithms for solving SCGs. 
The first is a gradient descent algorithm implemented using the ILD-based approach. In each iteration, it first runs the ILD process until reaching a sufficiently precise WE, before calculating the gradient of the leader's cost through AD. The leader's decision is then updated using the gradient information and sent back to the lower level to start another iteration. Although the algorithm has a structure similar to that of SAB methods, it promises better scalability and greater efficiency because the ILD-based differentiable program enables the use of AD.  However, the algorithm's commitment to reaching WE in each iteration presents a potential challenge: on large-scale and/or highly congested problems for which equilibration tends to take many iterations, it could create computational graphs too ``deep" to unroll quickly even with AD. The second algorithm is proposed to address this challenge by adding a heuristic twist. It takes an evolutionary rather than an equilibrium approach to interpreting the followers’ behaviors. That is, instead of anticipating the followers' response at equilibrium, the leader may only look ahead along the followers' evolution path for a few steps. With such  \emph{limited anticipation}, the leader need not wait until the followers settle at a new equilibrium to correct course. Instead, the two parties can \emph{co-evolve}, meaning they each update decisions simultaneously in every step of a shared evolution process.  Short-term anticipation limits the depth of the computational graph through which the gradient of the leader's objective is evaluated. Co-evolution, on the other hand, breaks the bilevel hierarchy and turns the solution process into a single loop.

Against numerous benchmark methods from the literature, we test the proposed algorithms on two classic applications, the continuous network design problem (CNDP) and the second-best congestion tolling problem (SCTP).  A wide range of networks, including a real-world network with up to nearly 100,000 origin-destination (OD) pairs, are used in these experiments.
We find that the first algorithm consistently delivers solutions as good as any benchmarks included in the experiments. The SAB method can match it in solution quality but falls far behind in scalability and numerical
stability, especially on large networks. The second algorithm closely and reliably tracks the performance of the first in solution quality despite the heuristic nature. Importantly, with a computation effort comparable to those of most mainstream heuristics, it is capable of providing solutions of high quality with greater regularity.
}

\subsection{Organization}

{The rest of the paper is organized as follows. Section \ref{sec:related} discusses related studies. Section \ref{sec:formulation} provides a general formulation for SCGs in the form of bilevel programs. Section \ref{sec:day-to-day} describes ILD and establishes conditions under which it converges to a WE. Section \ref{sec:ad-gradient} studies how to leverage AD to unroll ILD. The two proposed algorithms are presented in Sections \ref{sec:double} and \ref{sec:single-loop}. Results of numerical experiments are reported in Section \ref{sec:experiments}. Eventually, Section \ref{sec:conclusion} concludes the paper with a summary of the main results and future directions.}

\subsection{Notation}

We use $\sR$ and $\sR_+$ to denote the set of real numbers and non-negative real numbers.
For a vector $\va \in \sR^n$, we denote its $\ell_1$, $\ell_2$ and $\ell_{\infty}$ norms as $\|a\|_1$, $\|a\|_2$ and $\|a\|_{\infty}$ and its support as $\supp{(\va)} = \{i: \eva_i > 0\}$.
For two vectors $\va, \vb \in \sR^n$, their inner product is denoted as $\langle \va, \vb \rangle$. 
For a matrix $\mA \in \sR^{n \times m}$, we write $\text{nnz}(\mA)$ as the number of nonzero elements in matrix $\mA$. For a finite set $\gA$, we write $|\gA|$ as the number of elements in $\gA$ and $2^{\gA}$ as the set of all subsets of $\gA$. 

\section{Related Work}
\label{sec:related}

{ 

SCG has many applications across various disciplines. We motivate our study with two classic examples in transportation: (i) the continuous network design problem (CDNP) \citep{abdulaal1979continuous}, which seeks a socially optimal road expansion policy, and (ii) the second-best congestion tolling problem (SCTP), which aims to set tolls at selected locations in a network to minimize total travel delay \citep{verhoef2002second}.
The reader is referred to \citet{migdalas1995bilevel} and \citet{farahani2013review} for other common applications. %

In what follows, we confine our focus on related algorithms, starting with sensitivity analysis-based (SAB) methods (Section \ref{sec:sab}) and heuristic methods (Section \ref{sec:heuristics}) for SCGs, as the former's conceptual similarities with our algorithms and the latter's strong scalability make them natural benchmarks for comparison.  In Section \ref{sec:ml-bilevel}, we introduce the recent effort in developing AD-based bilevel programming algorithms.  }

\subsection{SAB algorithms for SCGs}
\label{sec:sab}

{

An SAB algorithm calculates the gradient of an SCG by examining the lower-level problem's sensitivity. With the ability to compute the gradient at any given feasible point, the algorithm iteratively descends towards a local stationary point along the direction of the negative gradient. It is worth noting that a stationary point is not necessarily a local optimum --- verifying that this is indeed the case is NP-hard  \citep{vicente1994descent}. Despite the theoretical difficulty in guaranteeing solution quality, SAB algorithms are widely considered a competitive SCG solution approach capable of delivering satisfactory solutions in practical applications \citep[see, e.g.,][]{yang1994traffic,chiou2005bilevel,zhang2018mitigating,zhang2019pool}.

There exist a variety of sensitivity analysis methods for the WE problem \citep[see, e.g.,][Chapter 4, for an overview]{yang2005mathematical}. The most popular method operates on the route-based variational inequality problem (VIP) formulation of the WE problem. When the WE problem has a unique route flow solution,   sensitivity may be obtained by performing {implicit differentiation} on the KKT conditions \citep{tobin1986sensitivity} or the fixed-point condition \citep{dafermos1988sensitivity}. One way or the other, a square matrix must be inverted whose size grows \textit{quadratically} with the size of the decision variable of VIP (controlled by the number of used routes at WE).

In networks of practical sizes, the WE problem usually has infinitely many route flow solutions \citep{sheffi1985urban}. The lack of uniqueness complicates the analysis because a ``representative" solution must be selected to properly define sensitivity.  
However, from the point of view of solving SCGs, it matters little which route flow solution is chosen as the representative provided the link flow solution is unique, which can be guaranteed under mild conditions. 
Building on \citet{tobin1986sensitivity}'s approach, \citet{tobin1988sensitivity} and \citet{yang2007sensitivity} each proposed a criterion for choosing a representative WE route flow.
While these approaches overcame the difficulty of non-unique route flow solutions,  their reliance on implicit differentiation remains a limiting factor for large-scale applications. 
Alternatively, \citet{patriksson2004sensitivity} suggested the curse of non-uniqueness in route flows can be bypassed by directly working on a link-based formulation. They demonstrate that the directional derivative of a (unique) WE link flow can be formulated as a distinct WE problem. Consequently, this derivative can be computed by solving the newly formulated WE problem. There is a caveat, however.  While the cost of computing a directional derivative at a specific direction is lower, getting the gradient typically requires repeating the calculation many times, one for each Cartesian coordinate in the feasible space. Moreover, computing a directional derivative still requires solving a WE problem, a computationally demanding task for large-scale applications.  This observation casts doubts on the overall performance of the method, especially when the upper-level problem has a large number of decision variables  \citep{yang2005mathematical}.

}

\subsection{Heuristic methods for SCGs}
\label{sec:heuristics}

{ 
Another class of methods for SCGs, broadly classified as ``heuristics" and popular for their scalability, tries to solve an easier problem whose solution is thought to approximate that of the SCG reasonably well. These heuristics are usually problem-specific, meaning their applicability may be restricted to problems with certain structures, and their solution quality may vary widely with problem settings. Given the vast literature on the topic, we confine our focus to a few popular heuristic algorithms for solving the problems that chiefly concern this study, namely CNDPs and SCTPs. 

For CNDPs, the iterative optimization-assignment (IOA) algorithm \citep{tan1979hybrid}, which iterates between finding the best response of the leader and that of the followers while fixing the decisions of the other player(s), is a practical solver \citep{universite1981design,friesz1985properties, marcotte1986network, marcotte1992efficient}. As it eliminates the dependency of the lower-level solution on the upper-level decision variables, IOA bypasses the need to ``differentiate through" the lower-level problem.
However, IOA does not solve the original Stackelberg game. Instead, as pointed out by \citet{fisk1984game} and \citet{friesz1985properties}, its underlying model is actually a Cournot game between the leader and the followers since the fixed point reached by IOA is a Nash-Cournot equilibrium from which neither player is incentivized to deviate. \citet{dantzig1979formulating} simplified a CNDP by assuming the leader and the followers work cooperatively to achieve the same objective. This leads to a single-level system optimal (SO) problem whose solution is taken as an approximate solution to the original CNDP.
There are two related schemes based on similar ideas, though their applicability requires a ``start-from-scratch" provision, which dictates the capacities of all links be allowed to be set to any number between 0 and infinity  \citep{marcotte1986network}. The first variant, proposed by \citet{marcotte1986network} and later dubbed the ``bring-to-equilibrium" (BTE) algorithm \citep{gairing2017complexity},  seeks to find the network capacity that induces the network flow pattern obtained from solving the SO problem. The other scheme, known as the uniform scaling algorithm \citep{gairing2017complexity}, scales the link capacities corresponding to the SO solution with a properly selected factor.
For CNDPs that satisfy the start-from-scratch requirement, \citet{marcotte1986network} proved the worst-case performance of IOA, SO, and BTE is linked to the price of anarchy (PoA) of the lower-level congestion game: the smaller the PoA, the better the approximation. \citet{gairing2017complexity} further showed that, for the same setting, the uniform-scaling algorithm and BTE achieve the same worst-case approximation guarantee, though a strictly better approximation guarantee can be achieved by consistently choosing the better of the solutions obtained by the two algorithms. 

Of the many heuristics developed for SCTPs, we focus on three recent additions \citep{harks2015computing}: the marginal cost tolling (MCT) algorithm, the exponential marginal cost difference tolling (EMCDT) algorithm, and the combinatorial tolling (CT) algorithm. MCT and EMCDT initialize tolls according to the {first-best} toll (or the marginal cost toll) and subsequently adjust them by comparing the WE solution induced by the current toll with the SO solution. CT starts from a no-toll solution and gradually increases link tolls, guided by the first-best tolls.   At each iteration of all three algorithms, the most costly task is computing the WE solution, for which highly efficient procedures exist.

}

\subsection{AD-based bilevel programming methods}
\label{sec:ml-bilevel}

Bilevel programming has found interesting applications in ML lately, ranging from hyperparameter optimization \citep{franceschi2018bilevel} and model-agnostic meta-learning \citep{finn2017model} to adversarial learning and neural architecture search \citep{liu2018darts}.
These ML-inspired bilevel programs typically have in their lower level a deep learning problem and are often solved by gradient descent algorithms, in which the gradient is evaluated either through implicit differentiation or AD \citep[see, e.g.,][for a survey]{liu2021investigating}.

{
The AD-based approach calculates the gradient of such bilevel programs by unrolling the deep neural network (DNN) training process, which technically solves the lower-level problem.}
In doing so, AD may be carried out fully  \citep[e.g.,][]{franceschi2017forward} or partially \citep[e.g.,][]{shaban2019truncated}. In both cases, the lower-level learning problem is solved exactly at first.  The difference between the full and the partial unrolling lies in the ``backward propagation" phase. Whereas full AD always unrolls the entire training process, partial AD ``truncates" it, i.e.,  unrolls only the tail of the process.   Truncation improves efficiency but may sacrifice accuracy. 
Pushing the idea of truncation to extreme yields a method called one-stage AD, which solves upper and lower levels in a single loop \citep{liu2018darts}. Specifically, whenever the trainable parameters in the learning problem are adjusted in one descent step, the algorithm unrolls it to obtain a gradient, which is then employed to update the upper-level decisions. Although the method performs well on many tasks \citep[e.g.,][]{luketina2016scalable, metz2016unrolled, finn2017model, liu2018darts}, its analytical properties have not been fully understood.

Inspired by the AD-based approach in ML, \citet{li2020end}   explored its application in SCGs.  { 
They represented the solution process for the lower-level WE problem by the Euclidean projection method \citep{bertsekas1982projection}, which equals solving a series of quadratic programs (QPs).  However, a QP cannot be directly unrolled by AD as it cannot be solved in closed form.  To overcome this hurdle, they turned to the CVXPYLayers package \citep{agrawal2019differentiable}. Instead of directly unrolling the QP, the package first solves the QP via CVX and then invokes implicit differentiation to obtain the gradient of its solution.  From the perspective of a user, CVXPYLayers casts a QP as a DiP and appears to automatically unroll the QP. Behind the scenes, however, the calculation done by implicit differentiation still suffers from the same scalability issue encountered by traditional SAB approaches such as those proposed by \citet{tobin1988sensitivity} and \citet{yang2007sensitivity} (see Section \ref{sec:sab}). Indeed, the computation effort needed to differentiate a single layer (iteration) of the Euclidean projection method is about the same as that of implicitly differentiating the lower-level WE.  As a result,  the algorithm proposed by \citet{li2020end} may be much less efficient than conventional SAB algorithms.

}

\section{Problem Formulation}
\label{sec:formulation}

We are now ready to describe the bilevel formulation for SCG.  
Since we are motivated mostly by transportation applications, we set the congestion game in a transportation network modeled as a directed graph $\gG(\gN, \gA)$, where $\gN$ and $\gA$ are the set of nodes and links, respectively. Accordingly, the followers in the SCG will be referred to as travelers hereafter. In the network, let $\gW \subseteq \gN \times \gN$ be the set of origin-destination (OD) pairs and $\gK \subseteq 2^{\gA}$ be the set of all routes. We use $\gK_w \subseteq \gK$ to denote the set of routes connecting $w\in\gW$ and $\gA_k \subseteq \gA$ the set of all links on the route $k \in \gK$. Also, $\emSigma_{w,k}$ denotes the OD-route incidence, with $\emSigma_{w,k} = 1$ if the route $k \in \gK_w$ and 0 otherwise, and $\emLambda_{a,k}$ denotes the link-route incidence, with $\emLambda_{a,k} = 1$ if $a \in \gA_k$ and 0 otherwise. For notational convenience, we write $\mLambda = (\emLambda_{a,k})_{a \in \gA, k \in \gK}$ and $\mSigma = (\emSigma_{w,k})_{w \in \gW, k \in \gK}$.

\subsection{Lower level: congestion game}
\label{sec:lower}

Let $\vd = (\evd_w)_{w \in \gW}$ be a vector with $\evd_w$ denoting the number of travelers between $w \in \gW$. The travelers' route choice is represented by a vector $\vp = (\evp_k)_{k \in \gK}$, where $\evp_k$ equals the \textit{proportion} of travelers selecting $k\in \gK_w$. The feasible region for $\vp$ can  be written as
$\gP = \{\vp \geq \vzero: \mSigma \vp = \vone\}$. Let $\vq = (\evq_k)_{k \in \gK}$ be a vector with $\evq_k = \evd_w$ if $k \in \gK_w$, and $\vf = (\evf_k)_{k \in \gK}$ and $\vx = (\evx_a)_{a \in \gA}$, with $\evf_k$ and $\evx_a$  be the number of travelers using route $k$ and link $a$, respectively. We have $\vq = \mSigma^{\T} \vd$, $\vf = \diag(\vq) \vp$, $\mSigma \vf = \vd$ and $\mLambda \vf = \vx$. The feasible region for $\vx$ can then be written as
$\gX = \{\vx \in \sR_+^{|\gA|}: \vx = \bar \mLambda \vp, \ \vp \in \gP\}$, where $\bar \mLambda = \mLambda \diag(\vq)$. { We further define $\vu = (\evu_a)_{a \in \gA}$ as a vector of link costs determined by a function $u: \gX \to \sR^{|\gA|}$.} Then, the vector of route cost  $\vc = \mLambda^{\T} \vu$. To summarize, the route cost function $c: \gP \to \sR^{|\gK|}$ can be defined as
\begin{equation}
    c(\vp) = \mLambda^{\T} \vu = \mLambda^{\T} u(\mLambda \vf) = \mLambda^{\T} u(\bar \mLambda\vp).
    \label{eq:calculate-c}
\end{equation}

A WE strategy can then be defined as follows \citep{wardrop1952road}.

{
\begin{definition}
\label{def:we}
A route choice strategy $\vp^* \in \gP$ is a WE strategy if $c_k(\vp^*) > b_w^* \Rightarrow \evp_k^* = 0$, where $b_w^* = \min_{k'  \in \gK_w} c_{k'} (\vp^*)$, for all $w \in \gW$ and $k \in \gK_w$ (i.e., a non-optimal route is used by no one).
\end{definition}
}

A WE strategy coincides with the solution to a variational inequality problem (VIP) \citep{dafermos1980traffic}. 

\begin{proposition}
\label{prop:vi-formulation}
A route choice $\vp^* \in \gP$ is a WE strategy if and only if $$\langle c(\vp^*), \vp - \vp^* \rangle \geq \vzero, \quad \vp \in \gP.$$
\end{proposition}

Denote $\gP^*$ as the set of WE strategies and $\gX^* = \{\vx \in \gX: \vx^* = \bar \mLambda \vp^*,  \vp^* \in \gP^*\}$ as the set of WE link flow.
The following two propositions \citep{dafermos1980traffic} characterize the geometry of $\gP^*$ and $\gX^*$. 
\begin{proposition}
\label{prop:solution-set-2}
If $c(\vp)$ is  strongly monotone { on $\gP$}, then $\gP^*$ is a singleton.
\end{proposition}
\begin{proposition}
\label{prop:solution-set}
If $u(\vx)$ is  strongly monotone { on $\gX$}, then $\gX^*$ is a singleton. Moreover, $\gP^*$ can be written as a polyhedron $\{\vp^* \in \gP: \bar \mLambda \vp^* = \vx^*\}$, where $\vx^*$ is the unique element in $\gX^*$.
\end{proposition}

{ Given that $\gX$ is a compact set, $u(\vx)$ is strongly monotone on $\gX$  as long as $\nabla u(\vx)$ is positive definite for all $\vx \in \gX$. When the link cost function is separable, i.e., the link cost $\evu_a = u_a(\vx)$ relies \emph{only} on $\evx_a$, this condition is satisfied when $u_a$ is strictly increasing with $\evx_a$, an assumption widely accepted in the transportation network modeling community. However, the strong monotonicity of $c(\vp)$, i.e., the positive definiteness of $\nabla c(\vp)$, does not hold in general networks even when $\nabla u(\vx)$ is positive definite.  Note that $\nabla c(\vp) =  \mLambda^{\T} \nabla u(\vx) \mLambda \diag(\vq)$. Thus, even when $\nabla u(\vx)$ is positively definite, $\nabla c(\vp)$ is positively definite if and only if $\mLambda$ has full column rank, which is rarely satisfied in real transportation networks \citep{sheffi1985urban}.
}

\subsection{Upper level: optimization problem}
\label{sec:upper}

The leader's decision is denoted as a vector $\vz \in \gZ \subseteq \sR^n$.  Since $\vz$ may influence link costs, the link cost function is parameterized as $u(\vx; \vz)$. The path cost function $c(\vp; \vz)$ and the set of WE strategies $\gP^*(\vz)$ are similarly parameterized.  { We assume the leader's cost be determined by the link flow $\vx \in \gX$ and its decision $\vz \in \gZ$, and represented by a \textit{continuously differentiable} function $l: \gX \times \gZ \to \sR$.  %
The leader's decision problem reads
\begin{equation}
\begin{split}
    \min_{\vz \in \gZ}~~&l(\vx^*; \vz), \\
    \text{s.t.}~~&\vx^* = \bar \mLambda \vp^*, \quad \vp^* \in \gP^*(\vz).
\end{split} \label{eq:bilevel}
\end{equation}

To simplify the analysis, we impose the following regulatory conditions.

\begin{assumption}
\label{ass:-1}
    The feasible region $\gZ$  is a convex set. 
\end{assumption}

\begin{assumption}
\label{ass:1}
The function $u(\vx; \vz)$ is twice continuously differentiable on $\gX \times \gZ$.
\end{assumption}

\begin{assumption}
\label{ass:1.5}
For all $\vz \in \gZ$, the function $u(\cdot; \vz)$ is strongly monotone on $\gX$.
\end{assumption}

By Proposition \ref{prop:solution-set}, Assumption \ref{ass:1.5} ensures that, for any $\vz \in \gZ$, all $\vp^* \in \gP^*(\vz)$ correspond to a unique $\vx^* \in \gX^*(\vz)$. This means there exists a well-defined implicit function mapping $\vz$ to the corresponding WE link flow $\vx^*$ and another mapping $\vz$ to the leader's cost $l(\vx^*; \vz)$.  Let us write these two implicit functions as $x^*(\vz)$ and $l^*(\vz)$. %
The following result, given by \citet{yang2007sensitivity}, establishes the differentiability of $x^*(\vz)$.

\begin{proposition}
\label{prop:diff}
    Under Assumptions \ref{ass:1}--\ref{ass:1.5}, for any $\vz \in \gZ$, if there exists $\vp^* \in \gP^*(\vz)$ that satisfies the strict complementary condition, i.e., $c_k(\vp^*) = \min_{k' \in \gK_w} c_{k'} (\vp^*) \Rightarrow \evp_k^* > 0$ for all $w \in \gK$ and $k \in \gK_w$ (no minimum-cost route is left unused),  $x^*(\vz)$ is differentiable at $\vz$. 
\end{proposition}
The next result, also due to \citet{yang2007sensitivity}, shows how to examine the sensitivity of WE, i.e., calculating $\nabla x^*(\vz)$, through an implicit differentiation-based method. 
\begin{proposition}
\label{prop:diff-formula}
Under Assumptions \ref{ass:1}--\ref{ass:1.5}, given $\vz \in \gZ$, suppose that $\bar \vp^* \in \gP^*(\vz)$ is a WE strategy that satisfies the strict complementary condition. 
Let $\gK^* \subseteq \gK$ correspond to a \textit{maximal set of independent columns} (MSIC) of the combined incidence matrix $[\mLambda_+^{\T}, \mSigma_+^{\T}]^{\T}$, where $\mLambda_+$ and $\mSigma_+$ are respective sub-matrices of $\mLambda$ and $\mSigma$ that contain only columns in $\supp(\bar \vp^*)$. Similarly, define $\mLambda_*$ and $\mSigma_*$ be respective sub-matrices of $\mLambda$ and $\mSigma$ that contain only columns in $\gK^*$. Let $\vf^* = \diag(\vq) \bar \vp^*$ and $\vx^* = \mLambda\vf^*$.  Then $\nabla x^*(\vz) = \mF^*_\vz \cdot \mLambda_*^{\T}$, where $\mF^*_\vz$ solves
    \begin{equation}
    \begin{bmatrix}
    \mLambda_*^{\T} \nabla_{\vx} u(\vx^*; \vz) \mLambda_* & -\mSigma_*^{\T} \\[3pt]
    \mSigma_* & \vzero
    \end{bmatrix}
    \cdot 
    \begin{bmatrix}
    \displaystyle 
    \mF^*_\vz \\[3pt]
    \mB^*_\vz
    \end{bmatrix}
    =
    \begin{bmatrix}
    -\mLambda_*^{\T} \nabla_{\vz} u(\vx^*; \vz)  \\[3pt]
    \vzero
    \end{bmatrix}.
    \label{eq:yang}
\end{equation}
\end{proposition}
Once $\nabla x^*(\vz)$ is obtained, the leader's decision can be improved through gradient descent, where the gradient $\nabla l^*(\vz) = \nabla_{\vz} l(\vx^*; \vz) + \nabla x^*(\vz) \cdot \nabla_{\vx} l(\vx^*; \vz)$. This is the basic logic of the SAB algorithm for SCGs.

{ Central to the sensitivity analysis method outlined above are two computationally intensive tasks: (i) finding an MSIC of the combined incidence matrix and (ii) solving the linear system given by Equation \eqref{eq:yang}. A common approach to Task (i) is to get a QR decomposition of the matrix and then extract MSIC from the non-zero diagonal elements of the resulting upper triangular matrix. To analyze the complexity of these tasks, define
$\textstyle \tau = \supp_{\vz \in \gZ}  \max_{\vp^* \in \gP^*(\vz)} |\supp(\vp^*)|/ |\gW|$ as the upper bound on the average number of equilibrium routes.  Thus, the total number of equilibrium routes is no more than $\tau \cdot |\gW|$. The dimension of the combined incidence matrix, as well as that of the linear system \eqref{eq:yang}, is then bounded below by $\mathcal{O}(\tau \cdot |\gW|)$. It follows that Tasks (i) and (ii)  both have a complexity of $\mathcal{O}(\tau^3 \cdot|\gW|^3)$.
On large regional transportation networks, $|\gW|$ is easily in the order of tens of thousands. At such a scale, these tasks can become prohibitively expensive.}

}

\section{Reformulation of WE Through ILD}
\label{sec:day-to-day}

At the core of our approach to SCG is the reformulation of WE as a DiP. In this section, we propose a formulation based on the (discrete-time) imitative logic dynamics (ILD) \citep{bjornerstedt1994nash}.  %
To describe the model, let us first define a parameterized function $h: \gP \times \gZ \to \gP$ such that
\begin{equation}
    h_k(\vp; \vz) = \frac{\evp_k \cdot \exp(-r \cdot c_k(\vp; \vz))}{\sum_{k' \in \gK_w} \evp_{k'} \cdot \exp(-r \cdot c_{k'}(\vp; \vz))}, \quad \forall k \in \gK_w, \quad \forall w \in \gW,
    \label{eq:def-h}
\end{equation}
{
where $r$ is a parameter (the dependency of $h_k$ on $r$ is omitted for simplicity). Starting from $h^{(0)}(\vp; \vz) = \vp$, we iteratively define $h^{(t + 1)}(\vp; \vz) = h(h^{(t)}(\vp; \vz); \vz)$ ($t = 0, 1, \ldots$). Given travelers' initial strategy $\vp^0 \in \gP$ at  $t = 0$ (interpreted as day 0), their strategy on day $t$ can be written as $\vp^{t} = h^{(t)}(\vp^0; \vz)$ according to ILD. We shall show that, with a proper $\vp^0$, the sequence $\vp^t$ defined above always converges to WE. This global convergence implies that we can obtain the derivative of the leader's cost at WE by differentiating the ILD, which is amenable to automatic differentiation thanks to the closed form.

In the following, we first interpret ILD as the outcome of all travelers simultaneously minimizing their expected costs via a special mirror descent method (Section \ref{sec:interpretation}). Based on this interpretation, Section \ref{sec:convergence-analysis} then establishes the conditions under which ILD converges to WE.
}

\subsection{Interpretation by mirror descent}
\label{sec:interpretation}

The mirror descent (MD) method was proposed by \citet{nemirovskij1983problem}. Similar to the projected gradient descent method, MD also seeks to diminish the ``distance" between two successive iterations.  The difference is that MD defines the distance using  \emph{Bregman divergence}, which is more general than the Euclidean distance used by the projected gradient descent method. 
\begin{definition}[Bregman divergence]
Let $\gY \subseteq \sR^{n}$ be a closed convex set and $\phi: \gY \to \sR \cup \{\infty\}$ be a { strictly} convex and continuously differentiable function. Then for all $\vy \in \gY$ and $\vy' \in \gY$, the Bregman divergence $D_{\phi}(\vy, \vy')$ induced by $\phi$ between $\vy$ and $\vy'$ is defined as
$D_{\phi}(\vy, \vy') = \phi(\vy) - \phi(\vy') - \left<\nabla \phi(\vy'), \vy - \vy'\right>.$
\end{definition}
Consider the problem of minimizing a continuous differentiable function $f: \gY \to \sR$ over a convex set $\gY \subseteq \sR^n$. At the current iteration $\vx^t \in \gY$, the MD method finds the next iterate $\vy^{t + 1}$ by solving
\begin{equation}
    \vy^{t + 1} = \argmin_{\vy \in \gY}~r \cdot  \langle \nabla f(\vy^t), \vy \rangle+ D_{\phi}(\vy, \vy^t).
    \label{eq:sub-mirror-descent}
\end{equation}
In a special case, the above problem admits an analytic solution \citep{beck2003mirror}. 
\begin{lemma}
\label{lm:ep}
     When $\gY = \{\vy \in \sR_+^{n}: \vone^{\T} \vy = 1\}$ is a probability simplex, specifying $\phi(\vy) = \langle \vy, \log \vy \rangle$ as the negative entropy function will lead to $D_{\phi}(\vy, \vy')=\langle \vy, \log \vy - \log \vy' \rangle$, the Kullback–Leibler (KL) divergence. In this case, the sub-problem \eqref{eq:sub-mirror-descent} of the MD method admits an analytic solution, which reads
    \begin{equation}
        \evy_k^{t + 1} = \frac{\evy_k^t \cdot \exp(- r \cdot \nabla f(\vy^t))}{\sum_{k = 1}^n \evy_k^t \cdot \exp(- r \cdot \nabla f(\vy^t))}.
        \label{eq:entropic-descent}
    \end{equation}
\end{lemma}

We are now ready to discuss how travelers' behavior underlying ILD can be interpreted through MD. We first note $\vp_w^t = (\evp_k^t)_{k \in \gK_w}$ can be viewed as the mixed strategy of a representative traveler from OD pair $w \in \gW$ on day $t$ constrained in $\gP_w = \{\vp_w \in \sR_{+}^{|\gK_w|}: \vone^{\T} \vp_w = 1\}$. Assume each representative traveler aims to minimize their expected cost $\langle \vc_w^t, \vp_w^t\rangle$, where $\vc_w^t = (c_k(\vp^t; \vz))_{k \in \gK_w}$ (note that  $\vp_w^t$ is the mixed strategy adopted by the representative traveler while  $\vp^t$ is a vector of ``stable" route choice probabilities).
As the contribution of any traveler's strategy to $\vp^t$ is { infinitesimal} in a nonatomic game,  $\vc_w^t$ is not affected by $\vp_w^t$. Thus, the derivative of the expected cost with respect to $\vp_w^t$ is simply $\vc_w^t$. Applying the MD method yields
\begin{equation}
    \min_{\vp_w \in \gP_w}~r \cdot \langle \vc_w^t, \vp_w \rangle + D_{\phi_w}(\vp_w, \vp_w^t), \quad \forall w \in \gW,
    \label{eq:mirror-descent-w}
\end{equation}
where $D_{\phi_w}(\vp_w, \vp_w^t)$ is the Bregman divergence induced by $\phi_w: \mathbb \gP_w \to \sR$. 
{
By applying Lemma \ref{lm:ep}, we can then directly obtain the following proposition. 

\begin{proposition}
\label{prop:ed-ild}
Letting $\phi_w(\vp_w) =  \langle \vp_w, \log \vp_w  \rangle$, given any  $\vp^t \in \gP$ and $\vp^{t + 1} = h(\vp^t; \vz)$,  $\vp_w^{t + 1} = (\evp_k^{t + 1})_{k \in \gK_w}$ solves the optimization problem  \eqref{eq:mirror-descent-w}.
\end{proposition}
}

Proposition \ref{prop:ed-ild} implies that ILD can be explained as the outcome of all travelers simultaneously minimizing their expected costs by a special MD method. To the best of our knowledge, this is a new result.

\subsection{Convergence analysis}
\label{sec:convergence-analysis}

{ We next show that ILD always converges to WE 
if the route cost function $c(\cdot; \vz)$ is cocoercive on $\gP$. Following \citet{marcotte1995convergence} (Proposition 2.1), we first give sufficient conditions for cocoercivity. 
\begin{proposition}
\label{prop:check-cocoercive}
    Given any $\vz \in \gZ$, suppose that (i) the function $u(\cdot; \vz)$ is monotone and twice continuously differentiable on $\gX$ and (ii) the matrix $\nabla_{\vx} u(\vx; \vz)^2 + (\nabla_{\vx} u(\vx; \vz)^2)^{\T}$ is positive semi-definite (p.s.d.). Then, there always exists $L_{\vz} \geq 0$ such that the route cost function $c(\cdot; \vz)$ is $1/4L_{\vz}$-cocoercive on $\gP$, i.e., 
    \begin{equation}
        \langle c(\vp'; \vz) - c(\vp; \vz), \vp' - \vp \rangle \geq 1/4L_{\vz} \cdot \|c(\vp'; \vz) - c(\vp; \vz) \|_2^2, \quad \text{for all}~\vp, \vp' \in \gP.
    \end{equation}
\end{proposition}
\begin{proof}
    The reader is referred to Appendix \ref{app:check-cocoercive}.
\end{proof}

The two conditions specified in Proposition \ref{prop:check-cocoercive} can be easily satisfied in transportation applications. The first  (monotonicity and twice continuous differentiability of the link cost function) is standard.  The second condition follows directly from the first when the matrix $\nabla_{\vx} u(\vx; \vz)$ is symmetric. In other words, when the link cost function is (i) separable or (ii) non-separable but symmetric, the first condition alone is sufficient to ensure the cocoercivity of $c(\cdot; \vz)$. Even for a link cost function that is neither symmetric nor separable, the second condition may still be fulfilled when, for instance,  the cost on any link is affected by its own flow more than the flows on all other links combined. 
}

We are now ready to present the main convergence result.  Let us first define the \emph{cover} of any  $\vp_w \in \gP_w$, denoted by $\gQ_w(\vp_w) \subseteq \gP_w$, as $\gQ_w(\vp_w) = \{\vp_w' \in \gP_w: \supp{(\vp_w')} \subseteq \supp{(\vp_w)} \}$.
Accordingly, the cover of $\vp \in \gP$ is defined as $\gQ(\vp)= \prod_{w \in \gW} \gQ_w(\vp_w)$.  If $\vp'\in \gQ(\vp)$, we say $\vp'$ is \emph{covered} by $\vp$, which means any route with positive choice probability in $\vp'$ must also have positive choice probability in $\vp$. The following lemma \citep{kullback1997information}  characterizes the property of the KL divergence.
\begin{lemma}
\label{lm:kl-finite}
For any $\vp_w, \vp'_w \in \gP_w$,  $D_{\phi_w}(\vp_w, \vp_w') < \infty$ if and only if $\vp_w'$ is covered by $\vp_w$.%
\end{lemma}

Lemma \ref{lm:kl-finite} implies that $\vp^t \in \gQ(\vp^0)$ for all $t \geq 0$, i.e., travelers will consider \emph{only}  the routes initially included in the choice set as they update their strategies.  It follows that $\vp^t$ would never converge to a WE if the intersection of $\gQ(\vp^0)$ and $\gP^*(\vz)$ is empty. Thus, to ensure convergence, $\gQ(\vp^0) \cap \gP^*(\vz)$ must be nonempty. Our main result below verifies that this necessary condition is also sufficient when other appropriate conditions are imposed.

\begin{theorem}
\label{thm:convergence}
 Suppose that $c(\cdot; \vz)$ is $c_{\vz}$-cocoercive on $\gP$ and choose $r < 2c_{\vz}$. { Given any $\vp^0 \in \gP$ such that $\gQ(\vp^0) \cap \gP^*(\vz) \neq \emptyset$, the sequence $\{\vp^t\}$ defined by  ILD iterate \eqref{eq:def-h} converges to  a fixed point $\hat \vp \in \gP^*(\vz)$.}  %
\end{theorem}
\begin{proof}
The reader is referred to Appendix \ref{app:main-proof}. 
\end{proof}

{
Theorem \ref{thm:convergence} thus establishes the convergence of ILD under relatively mild conditions, which lays the foundation for our proposed algorithms. %
As both ILD and MD have been extensively studied in the literature, one naturally wonders if Theorem \ref{thm:convergence} is a new result. The reader interested in this question is referred to Appendix \ref{app:explanation}, which explains the novelty of the result and clarifies its difference with a few closely related works \citep[e.g.,][]{marcotte1995convergence,krichene2015convergence, mertikopoulos2019learning}. 
}

\section{ILD as a Differentiable Program}
\label{sec:ad-gradient}

Having reformulated the WE of a congestion game as the limit of ILD, we proceed in this section to show ILD can be coded as a differentiable program (DiP), which will subsequently allow us to evaluate the gradient of the leader's cost using automatic differentiation (AD).

\begin{figure}[ht]
\vspace{-0.1in}
\begin{center}
\centerline{\includegraphics[width=0.8\columnwidth]{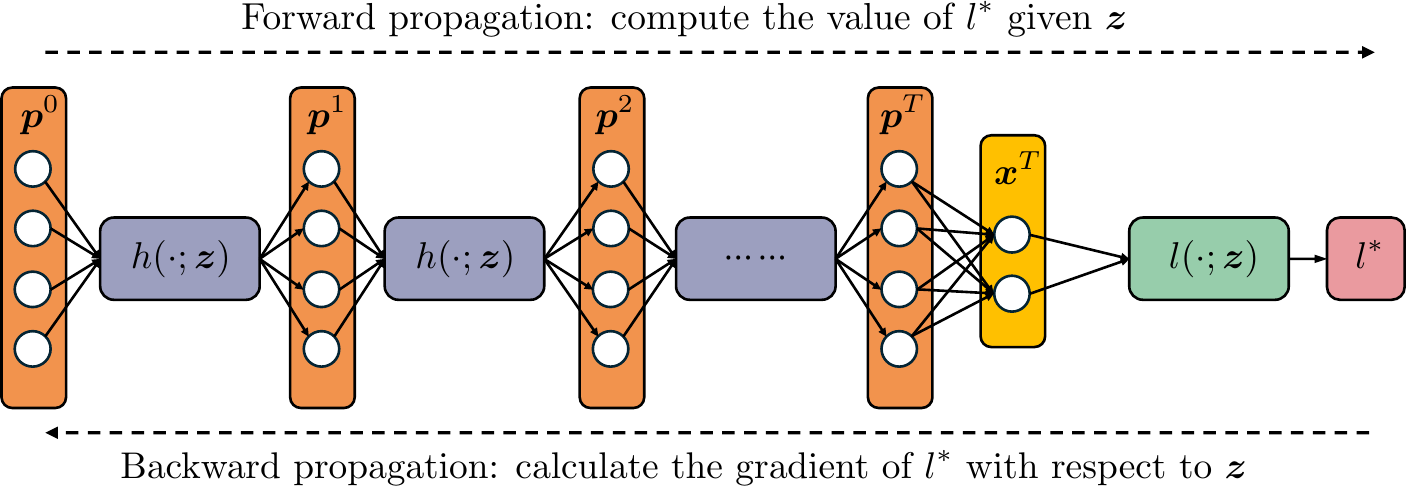}}
\caption{A DiP representation of the leader's cost in an SCG.}
\label{fig:framework}
\vspace{-0.1in}
\end{center}
\end{figure}

Recall that with ILD, we define travelers' strategy on day $t$ ($t = 0, 1, \ldots$) as $\vp^t = h^{(t)}(\vp^0; \vz)$.
Thus, the leader's cost on a given day, say $T$,  can be \emph{explicitly} expressed as $l^T = l(\vx^T; \vz)$, where $\vx^T = \bar \mLambda \vp^T$. It can be viewed as the output of a DiP structured as in Figure \ref{fig:framework}.
In the graph,  $h(\cdot; \vz)$ represents the hidden layers, and $l(\cdot; \vz)$  represents the output layer, to borrow the terminology from the deep learning community. In this ``DNN," the leader's decision $\vz$ is embedded as weights shared by all these layers. Then, calculating $\partial l^T / \partial \vz$ is reduced to computing the gradient of a DNN's output with respect to its weights. In deep learning, this ``routine" task is typically solved via the AD tools provided by deep learning frameworks, e.g., PyTorch \citep{paszke2019pytorch} and TensorFlow \citep{tensorflow2015-whitepaper}.

{ Similarly, AD tools can be leveraged to calculate $\partial l^T / \partial \vz$ in two phases. In the first --- or the forward propagation (FP) --- phase, $l^T$ is evaluated (see Algorithm \ref{alg:forward}), while all intermediate variables and their inter-dependencies are stored. Then, the second --- or the backward propagation (BP) --- phase unrolls all algorithmic operations in Algorithm \ref{alg:forward} in a reverse order, based on the chain rule defined by the inter-dependencies (see Algorithm \ref{alg:backward}).  Note that in Algorithm \ref{alg:backward}, a variable with a dot accent represents the gradient of $l^T$ with respect to the corresponding variable without the accent. To help the reader understand the code, take line \ref{line:e} in Algorithm \ref{alg:forward} for example. By applying the chain rule to this line, we then have
\begin{equation}
    \dot{\vc}^t = \frac{\partial l^T}{\partial \vc^t} = \frac{\partial \ve^t}{\partial \vc^t} \cdot \frac{\partial l^T}{\partial \ve^t} = - r\cdot \diag(\ve^t)  \dot \ve^t, \quad \text{where}~\ve^t = \exp(-r \cdot \vc^t),
\end{equation}
which corresponds to line \ref{line:e-c} in Algorithm \ref{alg:backward}. After all operations in Algorithm \ref{alg:forward} are unrolled, the gradients of $l^T$ with respect \textit{all} variables, including $\dot \vz = \partial l^T / \partial \vz$, are then derived. 

\begin{figure}[ht]
\begin{multicols}{2}
\begin{algorithm}[H]
\renewcommand{\algorithmicrequire}{\textbf{Input:}}
\renewcommand{\algorithmicensure}{\textbf{Output:}}
\caption{\small Forward propagation for $l^T$.}
\label{alg:forward}
\footnotesize
\begin{algorithmic}[1]
\For{$t = 0, \ldots, T - 1$}

    \vspace{1.5pt}

    \State {\color{black} $\vf^t = \diag(\vq)  \vp^t$} \Comment{$|\gK|$}
    \label{line:v-f-1}

    \vspace{1.5pt}
    \State {\color{black} $\vx^t = \mLambda \vf^t$} \Comment{$\text{nnz}(\mLambda)$}
    \label{line:m-f-1}

    \vspace{1.5pt}
    \State {\color{black} $\vu^t = u(\vx^t; \vz)$}
    \label{line:u}

    \vspace{1.5pt}
    \State {\color{black} $\vc^t = \mLambda^{\T} \vu^t$} \Comment{$\text{nnz}(\mLambda)$}
    \label{line:m-f-2}

    \vspace{1.5pt}
    \State {\color{black} $\ve^t = \exp(-r \cdot \vc^t)$} \Comment{$2|\gK|$}
    \label{line:e}

    \vspace{1.5pt}
    \State {\color{black} $\vq^t = \diag(\ve^t) \vp^t$} \Comment{$|\gK|$}
    \label{line:v-f-3}

    \vspace{1.5pt}
    \State {\color{black} $\vs^t = \mSigma^{\T}  \mSigma  \vq^t$} \Comment{$2 \text{nnz}(\mSigma)$}
    \label{line:m-f-3}

    \vspace{1.5pt}
    \State {\color{black}$\vp^{t + 1} = \vq^t / \vs^t$} \Comment{$|\gK|$}
    \label{line:v-f-4}
    \vspace{1.5pt}
\EndFor
\vspace{1.5pt}
\State {\color{black} $\vf^T = \diag(\vq) \vp^T$} \Comment{$|\gK|$}
\vspace{1.5pt}
\State {\color{black} $\vx^T = \mLambda \vf^T$} \Comment{$\text{nnz}(\mLambda)$}
\vspace{1.5pt}
\State $l^T = l(\vx^T; \vz)$
\label{line:l}
\end{algorithmic}
\end{algorithm}

\begin{algorithm}[H]
\renewcommand{\algorithmicrequire}{\textbf{Input:}}
\renewcommand{\algorithmicensure}{\textbf{Output:}}
\caption{\small Backward propagation for $\dot \vz = \partial l^{T}/\partial \vz$.}
\label{alg:backward}
\footnotesize

\begin{algorithmic}[1]
\State $\dot \vz = \nabla_{\vz} l(\vx^T; \vz)$ and $\dot \vx^T = \nabla_{\vx} l(\vx^T; \vz)$
\label{line:z-l}
\vspace{1.5pt}
\State {\color{black} $\overbar \vf^T = \mLambda^{\T}  \overbar \vx^T$} \Comment{$\text{nnz}(\mLambda)$}
\vspace{1.5pt}
\State {\color{black} $\overbar \vp^T = \diag(\vq) \overbar \vf^T$} \Comment{$|\gK|$}
\vspace{1.5pt}
\For{$t = T - 1, \ldots, 0$}
    \vspace{1.5pt}
    \State {\color{black}$\dot \vq^t = \vp^{t + 1}/ \vs^t$ and $\dot \vs^t =  -\dot \vp^t / (\vs^t)^2$} \Comment{$3|\gK|$}
    \label{line:v-b-1}
    \vspace{1.5pt}
    \State {\color{black}$\dot \vq^t = \mSigma^{\T} \, \mSigma  \, \dot \vs^t$} \Comment{$2\text{nnz}(\mSigma)$}
    \label{line:m-b-1}

    \vspace{1.5pt}
    \State {\color{black}$\dot \vp^t = \diag(\ve^t)  \dot \vq^t$} and {\color{black}$\dot
    \ve^t = \diag(\vp^t) \dot \vq^t$} \Comment{$2|\gK|$}
    \label{line:v-b-4}

    \vspace{1.5pt}
    \State {\color{black}$\dot \vc^t = -r \cdot \diag(\ve^t)  \dot \ve^t$}
    \Comment{$2|\gK|$}
    \label{line:e-c}

    \vspace{1.5pt}
    \State {\color{black} $\dot \vu^t = \mLambda  \dot \vc^t$} \Comment{$\text{nnz}(\mLambda)$}
    \label{line:m-b-2}

    \vspace{1.5pt}
    \State {\color{black} $\dot \vz \pluseq \nabla_{\vz} u(\vx^t; \vz)  \dot \vu^t$} and  {\color{black} $\dot \vx^t = \nabla_{\vx} u(\vx^t; \vz) \dot \vu^t$}
    \label{line:x-u}

    \vspace{1.5pt}
    \State {\color{black} $\dot \vf^t = \mLambda^{\T}  \dot \vx^t$} \Comment{$\text{nnz}(\mLambda)$}
    \label{line:m-b-3}

    \vspace{1.5pt}
    \State {\color{black} $\dot \vp^t \pluseq \diag(\vq) \dot \vf^t$} \Comment{$|\gK|$}
    \label{line:v-b-6}
\vspace{1.5pt}
\EndFor
\end{algorithmic}
\end{algorithm}
\end{multicols}
\vspace{-0.3in}
\end{figure}
Importantly, the implementation of the BP phase is effortless with the modern AD tools since Algorithm \ref{alg:backward} can be programmed and executed ``automatically" once the FP phase (Algorithm \ref{alg:forward}) is implemented.

}

{
We proceed to analyze the complexity of Algorithms \ref{alg:forward} and \ref{alg:backward}, which involve
two types of operations.  Type I operations, including Lines \ref{line:u} and \ref{line:l} in Algorithm \ref{alg:forward} and  Lines \ref{line:z-l} and \ref{line:x-u} in Algorithm \ref{alg:backward}, depend on the specification of $\gZ$ and $u(\vx; \vz)$.
In most SCG applications considered herein, $\vz$ contains design parameters on links, e.g., capacity added in CNDPs and toll levied in  SCTPs. Accordingly,  the complexity of a Type I operation is largely proportional to  %
$|\gA|$.  Below, we shall show the complexity of Type II operations scales with $|\gK|$, i.e., the number of routes.  As $|\gA| \ll |\gK|$ in practical networks, Type I operations may be ignored when analyzing the overall complexity.

Type II operations may take one of two basic forms.} (i) Element-wise vector operations (e.g., multiplication, division, and exponentiation): the complexity of these operations is determined by the size of the vector. (ii) Multiplication between a 0-1 sparse matrix and a vector: its complexity is determined by the number of nonzero elements in the matrix.  For example, the OD-route incidence matrix $\mSigma$ is a sparse matrix with only one nonzero element in each column.  Thus, $\text{nnz}(\mSigma) = |\gK|$.  As for the link-route incidence matrix $\mLambda$,  the number of nonzero elements $\text{nnz}(\mLambda) = \overbar{N}_{\text{link}} \cdot |\gK|$, where $\overbar{N}_{\text{link}}$ is the average number of links in a route. In Algorithms \ref{alg:forward} and \ref{alg:backward}, the estimated number of arithmetic operations is provided at the end of each line. Adding these numbers up, the overall number of arithmetic operations of the FP and BP algorithms, denoted respectively as $N_{\text{f}}$ and $N_{\text{b}}$, is  given by
\begin{align*}
        N_{\text{f}} &= (7T + 1) \cdot |\gK| + (2T + 1) \cdot \overbar{N}_{\text{link}} \cdot |\gK| \quad \text{and} \quad N_{\text{b}} &= (10T + 1) \cdot |\gK| + (2T + 1) \cdot  \overbar{N}_{\text{link}} \cdot |\gK|.
\end{align*}
As $\overbar{N}_{\text{link}} \ll |\gA|$,  it is usually a relatively small number, likely well below  50, even in regional-scale networks \citep[see, e.g.,][]{xie2019new}. By assuming $\overbar{N}$ to be well-bounded for simplicity, we then have $N_{\text{f}} = N_{\text{b}} = \gO(T \cdot |\gK|)$. Moreover, it is easy to verify that $N_{\text{b}}/N_{\text{f}}$ is bounded from above by $(2 + 10)/(2 + 7) = 4/3$. 
This result echos \citet{griewank1989automatic}, who showed that in the practical application of AD, the computational cost of BP is usually no more than 1.5 times that of FP in practice.

To summarize, when differentiating the leader's cost on day $T$ by treating ILD as a DiP,  (i)  the computational cost basically scales with both $T$ and $|\gK|$ at a linear rate and (ii) the computational cost required in backward propagation is tightly bounded by that in forward propagation. %

\section{A Double-Loop Algorithm}
\label{sec:double}

In this section, we devise our first algorithm for SCGs.  By replacing the equilibrium constraint $\vp^* \in \gP^*(\vz)$ in the original formulation \eqref{eq:bilevel} with its DiP counterpart,  we arrive at the following formulation:
\begin{equation}
\begin{split}
    \min_{\vz \in \gZ}~~&l(\vx^*; \vz), \\
    \text{s.t.}~~&\vx^* = \bar \mLambda \vp^*, \quad \vp^* = \lim_{T \to \infty} h^{(T)}(\vp^0; \vz),
    \label{eq:reformulation}
\end{split}
\end{equation}
where $\vp^0 \in \gP$ is travelers' initial strategy. {

\begin{assumption}
    \label{ass:5-1}
    The initial strategy $\vp^0$ is chosen such that $\gQ(\vp^0) \cap \gP^*(\vz) \neq \emptyset$ for all $\vz \in \gZ$.
\end{assumption}

\begin{assumption}
    \label{ass:5-2}
    There exists $c > 0$ such that $c(\cdot; \vz)$ is $c$-cocoercive on $\gP$ for all $\vz \in \gZ$.
\end{assumption}

With the above two assumptions, Theorem \ref{thm:convergence} guarantees $\vp^* = \lim_{T \to \infty} h^{(T)}(\vp^0; \vz) \in \gP^*(\vz)$ if $r$ is set to be a sufficiently small constant such that $r \leq 2c$. Moreover, per Assumption \ref{ass:1.5} imposed earlier, the value of $l(\vx^*; \vz)$, where $\vx^* = \bar \mLambda \vp^*$, is the same for all $\vp^* \in \gP^*(\vz)$. Taken together, choosing $r \leq 2c$ ensures the equivalence between Problems \eqref{eq:reformulation}  and \eqref{eq:bilevel}.

The new, ILD-based formulation \eqref{eq:reformulation}  has two advantages over the original one. {First}, differentiating the objective function with respect to $\vz$ is easier.  Since $\vx^*$ is now defined explicitly through ILD rather than implicitly by the WE conditions, expensive implicit differentiation, as customary in SAB methods, can be avoided.  Instead, the explicit ILD formulation enables direct differentiation using AD, which, as we shall see, significantly improves scalability. {Second}, when the uniqueness of WE cannot be secured, the equilibrium ``chosen" by the new formulation is the limit of an evolutionary dynamical process, which can be interpreted as the most likely outcome corresponding to a specific choice behavior. This feature is especially convenient when different WE strategies do not always correspond to the same cost for the leader. In such a case, the new formulation obtains a unique solution (i.e., the limit of the ILD process) without having to impose additional criteria (see Appendix \ref{app:extension} for more discussions).}

We propose a double-loop mirror descent algorithm  (DolMD) (see Algorithm \ref{alg:stackelberg}) for solving Problem \eqref{eq:reformulation}:  the inner loop over $t$ solves the lower-level equilibrium problem, whereas the outer loop over $i$ optimizes the leader's decision.
\begin{algorithm}[ht]
   \caption{A Double-loop MD (DolMD) algorithm for solving Problem \eqref{eq:reformulation}.}
   \label{alg:stackelberg}
\begin{algorithmic}[1]
{\footnotesize
   \State {\bfseries Input:} $\vp^0 \in \gP$, $\vz^0 \in \gZ$, upper-level step size $\rho > 0$, lower-level step size $r > 0$, equilibrium threshold value $\varepsilon > 0$.
   \vspace{1.5pt}
   \For{$i = 0, 1, \ldots$}
   \vspace{1.5pt}
    \State {\bfseries FP} (running Algorithm \ref{alg:forward} until convergence):
    \vspace{1.5pt}
    \For{$t = 0, 1, \ldots$}
        \State Run $\vp^{t + 1} = h(\vp^t; \vz^i)$. If $\delta(\vp^t; \vz^i) \leq \varepsilon$, break and set $T = t$ and $l^T = l(\bar \mLambda \vp^T; \vz)$.
        \vspace{1.5pt}
    \EndFor
    \vspace{1.5pt}
    \State {\bfseries BP}: Calculate $l_{\vz} = \partial l^T / \partial \vz_i$ by unrolling FP via AD (AD tools automatically programmed and executed Algorithm \ref{alg:backward}).
    \vspace{1.5pt}
    \State {\bfseries Update the leader's decision}: set $\vz^{i + 1} = \argmin_{\vz \in \gZ}~\rho \cdot \langle l_{\vz}, \vz - \vz^{i} \rangle + D_{\psi}(\vz, \vz^i)$.
    \vspace{1.5pt}
    \State If $\vz^{i}$ converges, break and set $\vz^* = \vz^i$.
   \EndFor
}
\end{algorithmic}
\end{algorithm}
In each inner loop, the leader first anticipates the travelers' best response by iterating $\vp^{t + 1} = h(\vp^t; \vz^i)$ until a sufficiently precise WE is achieved. The termination condition is described by a gap function $\delta: \gP \times \gZ \to \sR$, defined as
\begin{equation}
    \delta(\vp; \vz) = -\frac{\langle c(\vp; \vz), \vp' - \vp \rangle}{\langle c(\vp; \vz), \vp \rangle}, \quad \text{where}~\vp' \in \argmin_{\vp'' \in \gP}~\langle c(\vp; \vz), \vp''  \rangle.
    \label{eq:eq-gap}
\end{equation}
Here, the route choice $\vp'$ given by Equation \eqref{eq:eq-gap} is often known as the all-or-nothing assignment in the traffic assignment literature. When $\delta(\vp^t; \vz^i) \leq \varepsilon$ (i.e., $\vp^t$ is sufficiently accurate as a WE),  we terminate the inner loop and set $T = t$. The route choice solution $\vp^T$ is then accepted as the WE strategy for evaluating the leader's cost $l^T = l(\vx^T; \vz)$, where $\vx^T = \bar \mLambda \vp^T$. The gradient $l_{\vz} = \partial l^T / \partial \vz_i$ is directly computed via AD, which automatically programs and executes Algorithm \ref{alg:backward}. Then, $l_{\vz}$ is fed back to the outer loop to update $\vz^i$ to $\vz^{i + 1}$ via one mirror descent (MD) step. Specifically,
$$
    \vz^{i + 1} = \argmin_{\vz \in \gZ}~\rho \cdot \langle l_{\vz}, \vz - \vz^{i} \rangle + D_{\psi}(\vz, \vz^i),
$$
where the choice of the Bregman divergence $D_{\psi}: \gZ \times \gZ \to \sR$ is application specific. The resulting algorithm (Algorithm \ref{alg:stackelberg}) is dubbed the double-loop MD (DolMD) algorithm as it has a double-loop structure, and both its inner and outer loops are related to the MD method.

{
Assumption \ref{ass:5-1} requires $\vp^0$ is selected such that for all $\vz \in \gZ$, there exists $\vp^* \in \gP^*(\vz)$ whose support is covered by $\supp(\vp^0)$. To meet this requirement, one may simply enumerate all available routes in the network and then set $\vp^0$ by assigning a non-zero initial flow to each route.  Such a brute force approach is, of course, infeasible on networks of practice size. Nor can one know in advance which set can cover all equilibrium routes. This challenge may be tackled through a standard \textit{route generation} routine that gradually builds and maintains a set of routes identified as having the potential to be used at equilibrium. In the appendix, we provide the pseudocode of DolMD with such a routine (Algorithm \ref{alg:stackelberg-rg} in Appendix \ref{app:DolMD}). 

}

{

{ Both FP and BP phases in Algorithm \ref{alg:stackelberg} have a complexity of $\gO(T \cdot |\gK|)$ per our analysis in Section \ref{sec:ad-gradient}, where $T$ represents the number of ILD iterations required to converge to a sufficiently precise WE, determined by the threshold $\varepsilon$. This, however, is an overestimation, as the practical implementation of Algorithm \ref{alg:stackelberg} does not need to involve \textit{all} available routes.  Recall that the total number of equilibrium routes is no more than $\tau \cdot |\gW|$.
If we generate WE routes iteratively, the complexity can be reduced to $\gO(T \cdot \tau \cdot |\gW|)$. Our numerical experiments indicate that for most practical purposes, a satisfactory WE solution can be reached with no more than a few hundred iterations ($T<1000$).
Therefore, for large-scale instances (where $|\gW|$ is in the order of thousands or more), DolMD would hold a significant computational advantage over SAB methods, which has a complexity of $\mathcal{O}(\tau^3 \cdot |\gW|^3)$ (see Section \ref{sec:upper}).  Our numerical results will confirm this analysis (see Section \ref{sec:experiments}).}

}

\section{A Single-Loop Algorithm}
\label{sec:single-loop}

As mentioned in Section \ref{sec:related}, the significant computational challenges involved in solving SCGs exactly have led many to develop approximation schemes. Of these, the most widely tested is to approximate the Stackelberg game with a Cournot game \citep{fisk1984game}, 
{ 
\begin{equation}
\begin{cases}
    \displaystyle \tilde \vz \in \argmin_{\vz \in \gZ}~l(\tilde \vx; \vz), \quad \text{where}~\tilde \vx =   \bar \mLambda \tilde \vp, \\[10pt]
    \tilde \vp \in \gP^*(\tilde \vz).
\end{cases}
\label{eq:cournot-o}
\end{equation}
In essence, this scheme takes away from the leader the power of anticipating the traveler's best response, and forces it to make decisions without this advantage. Thus, the leader's competitive edge is weaker in the Cournot game than in the Stackelberg game. Nevertheless,  previous studies have found the scheme can provide good-quality approximation for some SCGs, especially CNDPs.

Here, we set out to take this idea one step further by introducing a new feature called \emph{limited anticipation}, which gives back to the leader a limited ability to anticipate the travelers' response.    Corresponding to limited anticipation is the following model
\begin{equation}
\begin{cases}
    \displaystyle \tilde \vz \in \argmin_{\vz \in \gZ}~l(\tilde \vx; \vz), \quad \text{where}~\tilde \vx = \bar \mLambda h^{(T)}(\tilde \vp; \vz), \\[10pt]
    \tilde \vp \in \gP^*(\tilde \vz),
\end{cases}
\label{eq:cournot}
\end{equation} 
where $T$ is a positive integer.
In the new model, when $\tilde \vp$ is given, the leader no longer directly minimizes their cost dictated by the travelers' strategy $\tilde \vp$ as in Model \eqref{eq:cournot-o}. Instead, it makes decisions while anticipating that the travelers would move $T$ steps along their evolutionary path described by ILD, hence changing their strategy to $h^{(T)}(\tilde \vp; \vz)$.

There are a couple of reasons why  Model \eqref{eq:cournot}, with its limited anticipation, is expected to outperform the classic Cournot approximation.
\begin{itemize}
\item First, the new model better approximates the original SCG given by Problem \eqref{eq:bilevel} than the Cournot model.      Whereas the leader minimizes $l^*(\vz) = l(x^*(\vz); \vz)$ in the SCG, it minimizes in the Cournot model a surrogate $g(\tilde \vp; \vz) := l(\bar \mLambda \tilde \vp; \vz)$, where $\tilde \vp$ is fixed.  In Model \eqref{eq:cournot}, the leader adopts a new surrogate
    \begin{equation}
        g^{(T)}(\tilde \vp; \vz) := l(\bar \mLambda h^{(T)}(\tilde \vp; \vz); \vz),
    \end{equation}
for a fixed $\tilde \vp$.  The new surrogate $g^{(T)}(\tilde \vp; \vz)$ is a better approximation of  $l^*(\vz)$ than $g(\tilde \vp; \vz)$, since ILD, destined to converge to WE, can bring $\bar \mLambda h^{(T)}(\tilde \vp; \vz)$ closer to $x^*(\vz)$ than $\bar \mLambda \tilde \vp$. %

\item Second,  Model \eqref{eq:cournot} has greater applicability than the Cournot model. In many SCG applications, the leader's cost function $l(\vx; \vz)$ may not explicitly rely on $\vz$. In SCTP, for example, the cost to be minimized is the total travel time, which only depends on link flows. While link tolls affect link flows  (i.e., the decision vector $\vz$)  in the lower-level problem, they do not directly contribute to the leader's cost.  In this situation, the leader's problem per the Cournot approximation becomes minimizing $l(\tilde \vx)$ ($\tilde \vx = \bar \mLambda \tilde \vp$) over $\vz \in \gZ$ --- an idle problem since $\vz$ does not even affect $l(\tilde \vx)$.
 Model \eqref{eq:cournot} is not subject to this peculiar limitation, because the leader's cost is  $l(\bar \mLambda h^{(T)}(\tilde \vp; \vz))$, which depends on $\vz$ as long as $T >0$.
 
\end{itemize}

 }

{ { %
To solve Problem \eqref{eq:cournot}, we note that it may also be viewed as a Cournot game played by the leader and the travelers, in which the leader's cost function is changed to $g^{(T)}(\vp; \vz)$.  The idea leads to a ``single-loop" heuristic for solving the original problem, as described in Algorithm \ref{alg:cournot}.  In each iteration $i$, the leader and the travelers move forward --- locally update their decisions to reduce their costs --- simultaneously.} On the one hand, the travelers advance along the  ILD trajectory from $\vp^i$ to $h^{(T)}(\vp^i; \vz^i)$. On the other hand,  the leader updates its current decision $\vz^i$ by one MD step, minimizing $g^{(T)}(\vp^i; \vz)$ while fixing $\vp^i$, leading to
$$
    \vz^{i + 1} = \argmin_{\vz \in \gZ}~\rho \cdot \langle l_{\vz}, \vz - \vz^{i} \rangle + D_{\psi}(\vz, \vz^t),
$$
where $l_{\vz} = \partial g^{(T)}(\vp^i; \vz^i) / \partial \vz^i$. AD is again invoked to calculate $l_{\vz}$.  Specifically, in the forward pass, Algorithm \ref{alg:forward} is called to evaluate $g^{(T)}(\vp^i; \vz^i)$, with $\vp^0$ being  replaced by $\vp^i$. Then, $l_{\vz}$ can be obtained by automatically executing Algorithm \ref{alg:backward} (the backward pass).
}

\begin{algorithm}[H]
   \caption{A single-loop MD (SilMD) algorithm for solving Problem \eqref{eq:cournot}}
   \label{alg:cournot}
\begin{algorithmic}[1]
{\footnotesize
   \State {\bfseries Input:} $\vp^0 \in \gP$, $\vz^0 \in \gZ$, upper-level learning rate $\rho$, lower-level learning rate $r$
   \For{$i = 0, 1, \ldots$}
    \vspace{1.5pt}
    \State {\bfseries FP}: Calculate $l^T = g^{(T)}(\vp^i; \vz^i)$ by calling Algorithm \ref{alg:forward}.
    \vspace{1.5pt}
    \State {\bfseries BP}: Calculate $l_{\vz} = \partial g^T / \partial \vz^i$ by calling Algorithm \ref{alg:backward} (implemented via AD).
    \vspace{1.5pt}
    \State {\bfseries Update both the leader's decision and the travelers' strategies}:
    \vspace{1.5pt}
    \State Set $\displaystyle \vz^{i + 1} = \argmin_{\vz \in \gZ}~\rho \cdot \langle l_{\vz}, \vz - \vz^{i} \rangle + D_{\psi}(\vz, \vz^t)$ and $\vp^{i + 1} = h^{(T)}(\vp^i; \vz^i)$.
    \vspace{1.5pt}
    \State If $\delta(\vp^i; \vz^i) \leq \varepsilon$  and $\vz^{i}$ converges, break and set $\tilde \vz = \vz^i$ and $\tilde \vp= \vp^i$.
   \EndFor
}
\end{algorithmic}
\end{algorithm}

Algorithm \ref{alg:cournot} promises significant computational advantages over Algorithm \ref{alg:stackelberg} thanks to co-evolution and limited anticipation.   {   In Algorithm \ref{alg:stackelberg}, the outer loop guides the gradient descent of the leader's decision, whereas the inner loop determines the travelers' best response. The leader does not update its decision until the lower-level congestion game reaches (a sufficiently precise) equilibrium.  Also, after the decision is updated, the congestion game starts over from scratch. This structure means the congestion game will be solved multiple times, and each time, the number of iterations required to reach a desirable precision is unknown ex-ante. As discussed earlier, the complexity of FP and BP operations (Algorithms \ref{alg:forward} and \ref{alg:backward}) is proportional to this number of iterations. Thus, if the congestion game takes more iterations to equilibriate, FP and BP operations will consume more computation time.  The problem goes beyond computation time. There is also a concern for storage because the number of iterations also determines the depth of the computational graph constructed for FP and BP operations.  Co-evolution in Algorithm \ref{alg:cournot}  means the leader and the travelers are allowed to adjust their actions simultaneously, according to their counterpart's previous response.  As in the Cournot model, there is no need to reach equilibrium in an inner loop.   Unlike in the Cournot model, however, the leader can foresee $T$ steps into the future based on the current strategy taken by the travelers.  This gives it a competitive edge expected to better approximate the position of the leader in a Stackelberg game. At the same time,  by capping the number of forward anticipation steps at a predetermined small $T$, Algorithm \ref{alg:cournot} also assuages the aforementioned concern on computational inefficiency.

{ It is worth reiterating that Algorithm \ref{alg:cournot} is a heuristic designed to find an approximate local solution. Its heuristic nature is underscored by the following observations:  (i) the algorithm solves an approximate model for the original SCG;  (ii) it aims to find a local solution to that model, and (iii) its convergence to a location solution has yet to be established rigorously (e.g., what combination of $\rho$ and $r$ can ensure convergence). In this study, we choose to demonstrate its performance through a comprehensive numerical study while leaving the theoretical analysis of convergence behavior and solution quality to future work.}

We close by noting Algorithm \ref{alg:cournot} can also be equipped with a standard {route generation} routine to gradually build a cover of WE routes. In Appendix \ref{app:DolMD}, we provide the pseudocode of SilMD with such a routine (see Algorithm \ref{alg:cournot-rg}). 
}

\section{Applications}
\label{sec:experiments}

{
We test the proposed algorithms in two classic SCG applications: the continuous network design problem (CDNP, Section \ref{sec:exp-cndp}) and the second-best congestion tolling problem (SCTP, Section \ref{sec:exp-sctp}).
The experiments are conducted on five networks, referred to for simplicity as Braess, Hearn, Sioux Falls, Barcelona, and Chicago Sketch. All have been used frequently in transportation literature. The topology of Braess and Hearn are shown in Figures \ref{fig:n1} and \ref{fig:n2}, respectively. For the other three,  Sioux-Falls has 528 OD pairs, Barcelona has 7865 OD pairs, and Chicago Sketch has 93512 OD pairs. The reader may consult the Transportation Networks GitHub Repository \citep{TransportationNetworks} for details. In all experiments, the travel time function takes the Bureau of Public Roads (BPR) form, i.e.,  $$u_{\text{time}, a}(\evx_a) = \evu_{a, 0} \cdot (1 + 0.15 \cdot (\evx_a / \evv_{a, 0})^{4}),$$ where $\evu_{a, 0}$ is free-flow travel time and $\evv_{a, 0}$ represents link capacity. %
\noindent
\begin{figure}[ht]
    \begin{minipage}[t]{0.45\textwidth}
    \centering
    \includegraphics[height=0.288\textwidth]{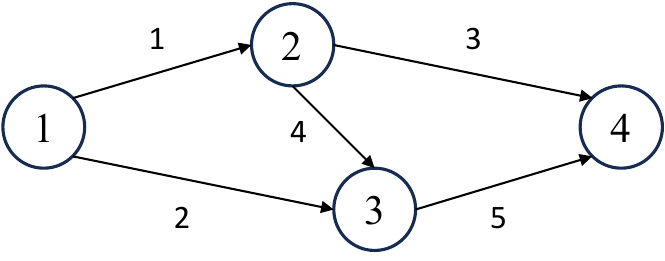}
    \vspace{0.025in}
    \captionof{figure}{Braess network.}
    \label{fig:n1}
    \end{minipage}
    \begin{minipage}[t]{0.54\textwidth}
    \centering
    \includegraphics[height=0.23\textwidth]{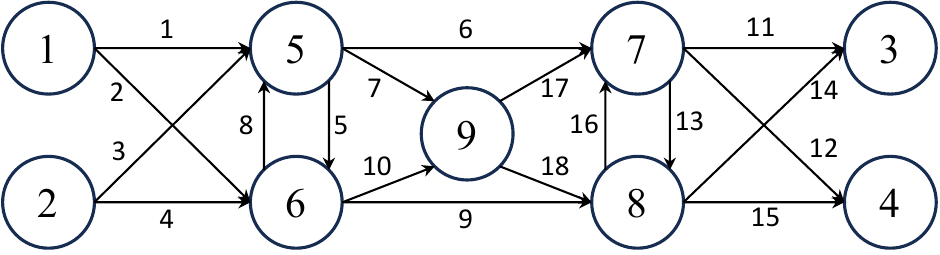}
    \vspace{0.025in}
    \captionof{figure}{Hearn network.}
    \label{fig:n2}
    \end{minipage}
    
    \vspace{-0.05in}
\end{figure}

We shall compare the proposed algorithms (DolMD and SilMD) with a wide range of alternatives from the literature, all implemented, coded, and tested in the same computation environment. { Both DolMD and SilMD are equipped with a route generation routine to build a cover of WE routes from scratch.} Appendix \ref{app:comparing-alg} provides the implementation details of all tested algorithms. 
Below, we outline the steps taken to ensure a fair comparison.

\begin{itemize}
    \item \textbf{WE calculation}. Many alternative algorithms require finding WE (i.e., solving a traffic assignment problem). Rather than ILD,  they may choose a more efficient traffic assignment algorithm. To this end, the improved gradient projection (iGP) algorithm, a highly efficient, route-based traffic assignment algorithm developed based on \citet{jayakrishnan1994faster} and \citet{xie2018greedy}, is employed. When the traffic assignment problem must be solved repeatedly, the iGP algorithm is initialized by a ``warm-start," i.e., starting from the WE solution obtained from the previous iteration. 
    \item \textbf{Equilibrium threshold}.
    Our algorithms need a threshold (i.e., the parameter $\varepsilon$ in Algorithms \ref{alg:stackelberg} and \ref{alg:cournot}) to determine whether the quality of the lower-level solution is adequate. The same requirement applies to any alternative algorithm for which finding WE is a subroutine. To ensure consistency, the same threshold value is used across all algorithms for the same test instance. %
    \item { \textbf{Convergence threshold}. In most cases, the solution process is terminated when $\|\vz^i - \vz^{i - 1} \|_{\infty} < \xi$, i.e., the difference between two successive solutions becomes sufficiently small.  For some algorithms, including SilMD, additional convergence criteria may be needed. 
    }

    \item \textbf{Objective function value.} 
    The leader's objective function value depends on  
    the decision solution vector $\vz^*$ and the corresponding WE link flow.  For any $\vz^*$, the WE link flow is always obtained with the same equilibrium threshold ($\varepsilon \le 10^{-6}$) across all algorithms.  %
\end{itemize}

Numerical results were produced on a Ubuntu 20.04.4 LTS workstation with 80 $\times$ Intel(R) Xeon(R) Gold 6242R CPUs and 10 $\times$ NVIDIA RTX A6000 iGPUs (CUDA version 12.4). %

\subsection{Continuous network design problems (CNDPs)}
\label{sec:exp-cndp}

In a CNDP, the leader is a network designer aiming to expand road capacity, whose decision $\vz \in \sR_+^{|\gA|}$ represents the capacity to be added to the current capacity $\vv_0 \in \sR_+^{|\gA|}$. After expansion, the cost for using link $a \in \gA$ becomes $u_a(\vx; \vz) = \evu_{0, a} \cdot (1 + 0.15 \cdot (\evx_a / (\evv_{0, a} + \evz_a))^4)$. We assume that expansion could only take place on a predetermined set of links $\tilde \gA \subseteq \gA$. Thus, $\evz_a = 0$ for all $a \in \gA \setminus \tilde \gA$. The cost of expansion is assumed to be $m(\vz) = \langle \vw, \vz^2 \rangle$ for some $\vw \in \sR_+^{|\gA|}$.
Accordingly, the leader's cost function can be written as $l(\vx; \vz) =  \langle \vx, u(\vx; \vz) \rangle + \beta \cdot m(\vz)$ for some $\beta > 0$, which is a weighted sum of total travel time (user cost) and expansion cost.

We test CNDPs on four networks: Braess, Sioux Falls, Barcelona, and Chicago Sketch.  For Sioux Falls, which is perhaps the most popular CNDP test case, ten out of 76 links can be expanded \citep{suwansirikul1987equilibrium}.  For Braess, expansion is allowed on all links.   On Barcelona and Chicago Sketch,  the number of links that can be expanded is set to 50. 

\vspace{0.05in}
\noindent
\textbf{Benchmark solutions.} For Braess, the cutting plane (CP) algorithm \citep{marcotte1983network} can solve the problem exactly (see Appendix \ref{app:cp} for more details). On Sioux-Falls, the best solution recorded in literature was given by the simulated annealing algorithm proposed by \citet{friesz1992simulated}. After 3900 objective evaluations, it reached an objective function value of $80.29$. We note that this value is obtained using our equilibrium threshold ($\varepsilon = 10^{-6}$); the value reported in their paper is $80.87$, which presumably corresponds to a less precise WE solution. 
We implemented the dual annealing (DA) algorithm  \citep{xiang1997generalized}, an algorithm that combines the advantages of classic simulated annealing and fast simulated annealing at a higher computation cost of 20,000 objective evaluations, which reached an objective function value of $79.90 < 80.29$. This solution is treated as the benchmark. For Barcelona and Chicago-Sketch, the benchmark solutions are also generated by our DA implementation, using 20,000 and 10,000 objective evaluations, respectively. %

\vspace{0.05in}
\noindent
\textbf{Alternative algorithms.} We consider the following algorithms: (1) \citet{yang2007stackelberg}'s SAB algorithm; see Appendix \ref{app:sab}. (2) Two classic CNDP heuristics: the iterative optimization-assignment (IOA) algorithm, which solves the approximation problem \eqref{eq:cournot-o} (Algorithm \ref{alg:ioa} in Appendix \ref{app:heuristics-cndp}); and the system-optimization (SO) algorithm, which solves the following convex program
\begin{equation}
    (\hat \vz, \hat \vp) \in \argmin_{\vz \in \gZ, \ \vp \in \gP} l(\vx; \vz), \quad \text{s.t.}~\vx = \bar \mLambda \vp,
\label{eq:so-main}
\end{equation}
and takes $\hat \vz$ as an approximate solution.  Problem \eqref{eq:so-main} is solved via a specialized algorithm (Algorithm \ref{alg:so} in Appendix \ref{app:heuristics-cndp}), which is significantly more efficient than using off-the-shelf solvers built in Python. %

\subsubsection{Braess}
It is well known expanding the capacity on Link 4 (the ``bridge" link in Figure \ref{fig:cndp-braess}) in Braess is counterproductive since it is bound to increase the total travel time at WE \citep{braess1968paradoxon}. Therefore, we expect the optimal capacity expansion scheme would forbid expanding that link. The benchmark solution given by the CP algorithm (the exact solution; see Table \ref{tab:cndp-braess}) confirms this intuition.  
}

{
In total, five algorithms, namely, IOA, SO, SAB, DolMD, and SilMD, are tested on Braess. For DolMD and SAB, we fix the convergence threshold $\xi = 10^{-5}$ and vary the equilibrium threshold $\varepsilon$ from $10^{-6}$ (high precision) to $10^{-3}$ (low precision). For SilMD, we vary $T$ from 1 to 5 while fixing $\xi = 10^{-5}$ and $\varepsilon = 10^{-7}$. 

Figure \ref{fig:cndp-braess} reports the optimality gaps of the solutions obtained from each algorithm, defined as the \textit{relative difference} between their corresponding objective values and the objective function value of the exact solution. Note that the curves representing IOA and SO are flat because, per design, their performance is not affected by $\varepsilon$ or $T$.   Table \ref{tab:cndp-braess} compares the objective function value, as well as the optimal values of new capacity ($\vz$), obtained by all algorithms. For DolMD and SilMD, the table reports only their worst solutions, obtained with the lowest equilibrium threshold for DolMD and with the smallest $T$ for SilMD.

\begin{figure}[ht]
\centering
\noindent
\vspace{0.1in}
\begin{minipage}[t]{0.54\textwidth}
    \centering
    \vspace{-1.15in}
    \captionof{table}{Solutions obtained by CP, IOA, SO, DolMD (with $\varepsilon = 10^{-3}$), and SilMD (with $T = 1$) together with their objective values in solving CNDP on Braess.}
    \scriptsize
    \vspace{0.1in}
    \begin{tabular}{cccccccc}
    \toprule
    Method & objective &       & $\evz_1$ & $\evz_2$  &$\evz_3$  & $\evz_4$  & $\evz_5$  \\
    \midrule
    CP    & 28.919 &       & 0.930 & 0.017 & 0.017 & 0     & 0.930 \\
    SAB   & 28.919 &       & 0.930 & 0.017 & 0.017 & 0     & 0.930 \\
    IOA   & 38.786 &       & 2.075 & 0     & 0     & 2.830 & 2.075 \\
    SO   & 29.275 &       & 0.825 & 0.030 & 0.030 & 0.113 & 0.825 \\
    DolMD-$10^{-3}$ & 28.921 &       & 0.944 & 0.016 & 0.016 & 0     & 0.944 \\
    SilMD-1 & 28.925 &       & 0.966 & 0.015 & 0.015 & 0     & 0.966 \\
    \bottomrule
    \end{tabular}%
  \label{tab:cndp-braess}%
    \end{minipage}
    \hfill
\noindent
    \begin{minipage}[t]{0.43\textwidth}
    \centering        \includegraphics[width=0.95\textwidth]{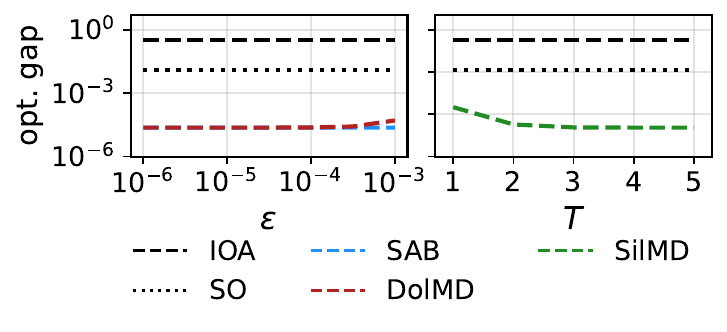}
    \captionof{figure}{Optimality gaps of the solutions obtained by IOA, SO, SAB, DolMD, and SilMD in solving the CNDP problem on Braess.}
    \label{fig:cndp-braess}
    \end{minipage}
\end{figure}

The SAB solution is almost identical to the exact solution and,  interestingly,  not affected much by $\varepsilon$ in this network. The IOA solution stands out in that it advises investing heavily on the wrong link, Link 4. Not surprisingly, it produces a total system travel time nearly 35\% worse than the optimal value. The SO solution is much better, though the investment on Link 4 remains noticeably above zero.  DolMD performs equally well as SAB when $\varepsilon < 10^{-3.5}$. For a coarser equilibrium threshold, DolMD's performance begins to degrade. Nevertheless, at $\varepsilon = 10^{-3}$, DolMD still obtains an objective function value of 28.9206, far better than either the IOA or the SO solution. As for SilMD, its performance stabilizes and matches that of SAB once $T \geq 2$. At $T = 1$, it performs slightly worse than  DolMD at  $\varepsilon = 10^{-3}$. As a sign of robustness, even the worst solutions given by DolMD and SilMD never invest a dim on the bridge link.

\subsubsection{Larger networks} 
For larger networks, no exact solutions are known.  To assess the quality of a given solution $\vz$, we scale $l^*(\vz)$ to $(l^*(\vz) - l_{\text{so}}) / l_{\text{so}}$ by using $l_{\text{so}} = l(\hat \vz; \hat \vp)$, i.e., the optimal objective function value of the solution to the SO problem \eqref{eq:so-main}, as a lower bound.   The first two lines of Table \ref{tab:cndp-result} report this relative objective function value for the all-zero solution $\vz = \vzero$ and the benchmark solution given by the DA algorithm. %

We start all algorithms from the initial solution $\vz^0 = \vzero$. For IOA, SAB, DolMD, and SilMD, the equilibrium threshold $\varepsilon$ is set to $10^{-5}$ on Sioux Falls and $10^{-4}$ on the other two. For SilMD, two values of $T$, 10 and 40, are tested. The results are reported in Table \ref{tab:cndp-result} and Figure \ref{fig:cndp}.

\begin{table}[htbp]
  \centering
  \caption{Performance of IOA, SO, SAB, DolMD, and SilMD (with $T = 10, 40$) in solving CNDPs on the Sioux-Falls, Barcelona, and Chicago-Sketch networks. The ``rel. obj." column indicates the relative objective value of the solutions; the ``time" and ``iters" columns record computation time and the number of iterations, respectively; the ``tpi" column reports the average computation time per iteration.}
  \vspace{0.05in}
  \scriptsize
    \begin{tabular}{cccccccccccccccc}
    \toprule
    \multirow{2}[4]{*}{Method} &       & \multicolumn{4}{c}{Sioux Falls} &       & \multicolumn{4}{c}{Barcelona} &       & \multicolumn{4}{c}{Chicago Sketch} \\
\cmidrule{3-6}\cmidrule{8-11}\cmidrule{13-16}          &       & rel. obj. & time (s) & iters & tpi (s) &       & rel. obj. & time (min) & iters & tpi (s) &       & rel. obj. & time (min) & iters & tpi (s) \\
    \midrule
    All-zero &       & 28.63\% &       &       &       &       & 4.33\% &       &       &       &       & 6.92\% &       &       &  \\
    \midrule
    DA &       & \textcolor[rgb]{ 1,  0,  0}{2.83\%} & 1436  & $20000$ & 0.072 &       & 3.00\% & 514   & $20000$  & 1.5   &       & 3.48\% & 2244  & $10000$  & 13 \\
    \midrule
    SAB   &       & \textcolor[rgb]{ 1,  0,  0}{2.83\%} & 117   & 336   & 0.35  &       & \multicolumn{4}{c}{\textcolor[rgb]{ .502,  .502,  .502}{Does not converge}} &       & \multicolumn{4}{c}{\textcolor[rgb]{ .502,  .502,  .502}{Exceeds memory limit}} \\[2pt]
    DolMD &       & \textcolor[rgb]{ 1,  0,  0}{2.83\%} & 19    & 64    & 0.30  &       & \textcolor[rgb]{ 1,  0,  0}{2.62\%} & 7.3  & 259   & 1.7   &       & \textcolor[rgb]{ 1,  0,  0}{3.39\%} & 81    & 474   & 10 \\
    \midrule
    IOA   &       & 3.50\% & 0.73  & 12    & 0.061 &       & 2.99\% & 0.22  & 6     & 2.2   &       & 3.43\% & 1.5   & 7     & 13 \\
    SO    &       & 3.23\% & 0.76  & 43    & 0.018 &       & 2.72\% & 1.8  & 69    & 1.6   &       & 3.41\% & 9.3   & 57    & 9.8 \\
    SilMD-10 &       & 3.23\% & 1.3   & 90    & 0.014 &       & 2.67\% & 1.4  & 297   & 0.28  &       & 3.40\% & 4.4   & 542   & 0.49 \\
    SilMD-40 &       & 2.97\% & 2.3   & 60    & 0.038 &       & 2.65\% & 2.2  & 355   & 0.38  &       & 3.40\% & 17    & 519   & 1.9 \\
    \bottomrule
    \end{tabular}%
  \label{tab:cndp-result}%
\end{table}%

\textbf{DA vs. DolMD.} Figure \ref{fig:cndp} reveals a consistent pattern exhibited by the DA algorithm across all test instances: the relative objective value declines fast in episodes separated by long phases of ``inaction."  For example, in  Chicago-Sketch, little progress was made after the first two minutes until about two hours later. The next 1000 minutes saw steady improvement, followed by a 21-hour period during which the solution remained virtually unchanged.  This behavior is consistent with what one would expect from a meta-heuristic like DA.   Across all three networks, DolMD always performs as well as DA, but with a tiny fraction of the latter's computation time (about one hundredth).  In the two larger networks (Barcelona and Chicago-Sketch), DolMD actually delivers notably better solutions. While we do not know how close these solutions are to global optima, the fact that DolMD beats a well-known meta-heuristic with regularity highlights its ability to achieve high-quality solutions at an affordable computation cost.

\begin{figure}[ht]
\vskip 0.1in
\centering

\begin{subfigure}[b]{0.32\textwidth}
\includegraphics[width=1\columnwidth]{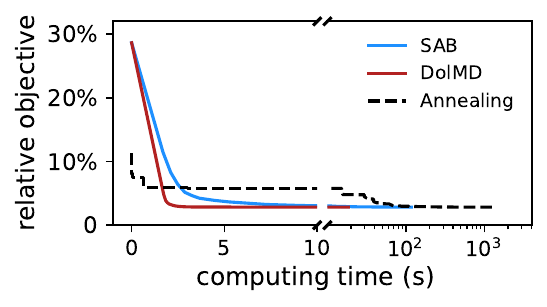}
\captionsetup{size=small}
\caption{Sioux Falls.}
\label{fig:cndp-sf}
\end{subfigure}
\begin{subfigure}[b]{0.32\textwidth}
\includegraphics[width=1\columnwidth]{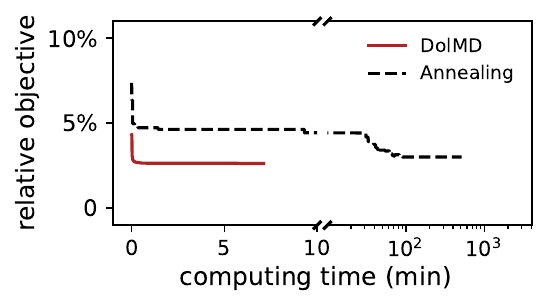}
\captionsetup{size=small}
\caption{Barcelona.}
\label{fig:cndp-bar}
\end{subfigure}
\begin{subfigure}[b]{0.32\textwidth}
\includegraphics[width=1\columnwidth]{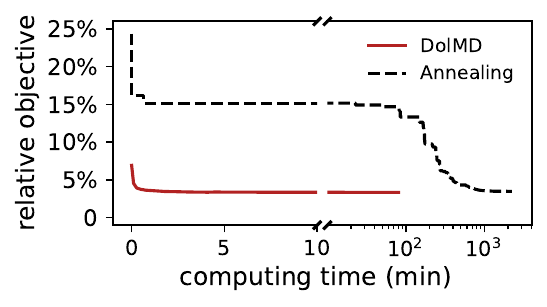}
\captionsetup{size=small}
\caption{Chicago Sketch.}
\label{fig:cndp-cs}
\end{subfigure}

\caption{The convergence curves of SAB (if the convergence is successful), DolMD, and the annealing algorithm in solving CNDPs on the Sioux-Falls, Barcelona, and Chicago-Sketch network.}
\label{fig:cndp}
\end{figure}

\textbf{SAB vs. DolMD.} 
For Sioux Falls, SAB and DolMD reached almost identical solutions that are as good as any solution obtained by the tested algorithms.  SAB required nearly five times as many iterations and six times as much computation time as DolMD. Its computation time per iteration is slightly higher than DolMD --- despite the use of iGP for solving WE --- likely because of the disadvantages inflicted by implicit differentiation (ID).  SAB's much-prolonged convergence is puzzling at first glance. 
The culprit, revealed by a close examination, turns out to be the identification of the maximum set of independent columns (MSIC), a key step in  \citet{yang2007sensitivity}'s algorithm. In our implementation, the original matrix is reduced by a QR decomposition to an upper triangular matrix, whose {non-zero} diagonal elements are then used to choose MSIC. However, this procedure may be numerically unstable since a column associated with a very small diagonal element may or may not belong to MSIC.  Ultimately, it is up to the modeler to decide ``how small can be deemed as zero."  We have tried different threshold values, but none have delivered a robust performance. Interestingly, this intricate problem did not affect SAB's ability to reach a high-quality solution on Sioux-Falls; it merely delayed convergence.

On Barcelona, DolMD converges smoothly, eventually reaching a solution much better than all competitors. In contrast, when the SAB algorithm was applied, the misidentification of MSIC became much more severe, causing the algorithm to diverge after five iterations. We subsequently tried the algorithm proposed by \citet{tobin1988sensitivity}, which does not require identifying MSIC. {Yet the algorithm still diverged (due to numerical errors of a different kind; see Appendix \ref{app:sab} for details) after 50 iterations, which lasted about 100 minutes. The long computation time per iteration indicates that even if the SAB algorithm can converge appropriately, it will take a significantly longer time than DolMD.}
On Chicago-Sketch, SAB became infeasible, as simply storing the linear system \eqref{eq:yang} is beyond the capability of our test environment.  DolMD not only worked but delivered the best solution after taking nearly 500 iterations and 1.5 hours of computation time.  To the best of our knowledge, this is the first time a first-order method has been applied to a CNDP of this scale.

\textbf{IOA, SO, and SilMD.} The performance of the two heuristics, IOA and SO, varies with the network. In Sioux Falls and Barcelona, the gap between them and DolMD is significant. For Chicago Sketch, however, the solution given by DolMD is only marginally better. Indeed, IOA and SO could provide nearly perfect solutions in some CNDP instances \citep[see, e.g., the numerical example in][]{marcotte1986network}.  In all three tests, IOA always runs faster and reaches a worse solution than SO, but either version of SilMD beats both IOA and SO in terms of solution quality. On Barcelona and Chicago Sketch, SilMD at $T = 10$ even outperforms SO in terms of computation time. %

\vspace{0.05in}
\noindent
\textbf{Summary.} DolMD and SAB tend to generate CNDP solutions of similarly high quality. However, DolMD is significantly more scalable and numerically stable.  IOA and SO performed surprisingly well in large CNDP instances, providing high-quality approximate solutions at a computation cost much lower than that required by DolMD. SilMD, on the other hand, seems to promise a more favorable balance between quality and efficiency than both DolMD and the heuristics.

\subsection{Second-best congestion tolling problems (SCTPs)}
\label{sec:exp-sctp}

In an SCTP, the leader is an infrastructure manager whose goal is to minimize the adverse impact of traffic congestion by choosing a vector of link toll, denoted as $\vz \in \sR_+^{|\gA|}$. With toll, the link cost function becomes $u_a(\vx; \vz) = u_{\text{time}, a}(\evx_a) + \lambda \cdot \evz_a$, where $1/\lambda$ is the traveler's value of time, assumed to be a constant across the population for simplicity. The toll can only be levied on a subset of links $\tilde \gA \subseteq \gA$, i.e., $\evz_a = 0$ for all $a \in \gA \setminus \tilde \gA$.  The leader's cost is the total travel time of all travelers, i.e., $l(\vx) = \langle u_{\text{time}}(\vx), \vx \rangle$. 
To evaluate the effectiveness of a tolling scheme, we compare the total travel time it induces with two reference points: $T_{\text{we}} = l(\vx_{\text{we}})$, where $\vx_{\text{we}}$ is the no-toll WE link flow pattern; and $T_{\text{so}} = l(\vx_{\text{so}})$, where $\vx_{\text{so}} = \argmin_{\vx \in \gX} l(\vx)$ is the system-optimal (SO) link flow pattern.
We thus gauge the effectiveness of a toll scheme $\vz$ by the \emph{relative excessive delay} (r.e.d.), computed by $(l^*(\vz) - T_{\text{so}}) / (T_{\text{we}} - T_{\text{so}})$. Clearly, the relative excessive delay must range between zero and one, and the closer to zero, the better.

We test SCTPs on Hearn and the three larger networks. 
In Hearn, we consider two settings: (i) only Links 11 and 12 are tolled; (ii) only Links 11, 12, and 2 are tolled. The Sioux-Falls instance, in which 18 links are tolled, is taken from \citet{lawphongpanich2004mpec}. For Barcelona and Chicago-Sketch,  we pick 20 and 40 links, respectively, for tolling.

\vspace{0.05in}
\noindent
\textbf{Benchmark solutions.} In Hearn, the decision variables are constrained in a low dimensional space so that global optima can be found via a brute-force grid search. { On Sioux-Falls,  the result given by the CP algorithm of \citet{lawphongpanich2004mpec} indicates that the optimal value is between  $72.6238$ (the upper bound)  and $72.1036$ (the lower bound); see Appendix \ref{app:cp} for how the lower bound is derived.} This corresponds to a relative excessive delay ranging between 5.64\% (the lower bound) and 23.87\% (the upper bound). Using the DA algorithm (with 32,000 objective evaluations), we located a solution with an objective value of $72.4216$ or a relative excessive delay of $16.78\%$. Since this solution is much better than the upper bound obtained by \citet{lawphongpanich2004mpec}, we take it as the benchmark. Similarly, for Barcelona and Chicago-Sketch, the DA algorithm, with 20,000 and 10,000 objective evaluations, respectively, is employed to produce a benchmark solution.

\vspace{0.05in}
\noindent
\textbf{Alternative algorithms}. We compare the proposed algorithms with three recent heuristic methods proposed by \citet{harks2015computing}: the marginal cost tolling (MCT) algorithm, the exponential marginal cost difference tolling (EMCDT) algorithm, and the combinatorial tolling (CT) algorithm. We refer the readers to Appendix \ref{app:h-sctp} for implementation details. Note that MCT is slightly modified to improve reliability. Hence, it is referred to as the revised MCT or rMCT hereafter.

\subsubsection{Hearn} 
The test results for Hearn are reported in Table \ref{tab:sctp-hearn}. According to the direct search method, the optimal tolling schemes can reduce the relative excessive delays to $97.6 \%$ and $73.1 \%$ respectively, in the first (Links 11 and 12 are tolled) and the second (Links 2, 11 and 12 are tolled) settings. This result indicates that tolling Link 2 is of high value to congestion reduction. For Setting (i), all three heuristics failed to identify a solution that can improve upon the status quo --- the relative excessive delay goes above one, implying the solution is worse than imposing no tolls. In Setting (ii), where tolling Link 2 is permitted, the heuristics performed much better, yielding solutions resembling the global optima reasonably well.

\begin{table}[ht]
  \centering
  \caption{Comparison of different algorithms in solving SCTP on two Hearn instances. The first column in each sub-table reports the global optima given by direct search.  The second column reports the results of SAB, DolMD and SilMD with an all-zero initial solution. Columns 3 to 5 report the results of the three algorithms with an initial solution given by one of the three heuristics (rMCT, EMCDT, and CT). }
  \label{tab:sctp-hearn}
    \vspace{-5pt}
    
    \begin{subtable}{0.99\textwidth}
  \footnotesize
  \centering
  \captionsetup{size=small}
  \caption{Links 11 and 12 are tolled.}
  \vspace{-2pt}
  \scriptsize
        \begin{tabular}{cccccccccccccccccc}
    \toprule
    Method & r.e.d.  &       & \multicolumn{2}{c}{Method} & r.e.d.  &       & \multicolumn{2}{c}{Method} & r.e.d.  &       & \multicolumn{2}{c}{Method} & rg    &       & \multicolumn{2}{c}{method} & r.e.d. \\
\cmidrule{1-2}\cmidrule{4-6}\cmidrule{8-10}\cmidrule{12-14}\cmidrule{16-18}          &       &       & All-zero &       & 100\% &       & rMCT  &       & 110\% &       & EMCDT &       & 108\% &       & CT    &       & 104\% \\
    Direct & \multirow{2}[0]{*}{\textcolor[rgb]{ 1,  0,  0}{97.6\%}} & \textcolor[rgb]{ 1,  0,  0}{} & \textcolor[rgb]{ .651,  .651,  .651}{(as $\vz^0$)} & \hspace{-0.1in} SAB   & \textcolor[rgb]{ 1,  0,  0}{97.6\%} &       & \textcolor[rgb]{ .651,  .651,  .651}{(as $\vz^0$)} & \hspace{-0.1in} SAB   & \textcolor[rgb]{ 1,  0,  0}{97.6\%} &       & \textcolor[rgb]{ .651,  .651,  .651}{(as $\vz^0$)} & \hspace{-0.1in} SAB   & 103\% &       & \textcolor[rgb]{ .651,  .651,  .651}{(as $\vz^0$)} & \hspace{-0.1in} SAB   & 103\% \\
     Search &       & \textcolor[rgb]{ 1,  0,  0}{} & $\drsh$ & \hspace{-0.1in} DolMD & \textcolor[rgb]{ 1,  0,  0}{97.6\%} &       & $\drsh$ & \hspace{-0.1in} DolMD & \textcolor[rgb]{ 1,  0,  0}{97.6\%} &       & $\drsh$ & \hspace{-0.1in} DolMD & \textcolor[rgb]{ 1,  0,  0}{97.6\%} &       & $\drsh$ & \hspace{-0.1in} DolMD & 103\% \\
          &       &       &       & \hspace{-0.1in} SilMD & \textcolor[rgb]{ 1,  0,  0}{97.6\%} &       &       & \hspace{-0.1in} SilMD & \textcolor[rgb]{ 1,  0,  0}{97.6\%} &       &       & SilMD & \textcolor[rgb]{ 1,  0,  0}{97.6\%} &       &       & \hspace{-0.1in} SilMD & \textcolor[rgb]{ 1,  0,  0}{97.6\%} \\
    \bottomrule
    \end{tabular}%
  \label{tab:5}%
  \end{subtable}
  
  \vspace{0.15in}
    \begin{subtable}{0.99\textwidth}
    \footnotesize
    \centering
    \captionsetup{size=small}
    \caption{Links 2, 11, and 12 are tolled.}
    \vspace{-2pt}
    \scriptsize
        \begin{tabular}{cccccccccccccccccc}
    \toprule
    Method & r.e.d.  &       & \multicolumn{2}{c}{Method} & r.e.d.  &       & \multicolumn{2}{c}{Method} & r.e.d.  &       & \multicolumn{2}{c}{Method} & rg    &       & \multicolumn{2}{c}{method} & r.e.d. \\
\cmidrule{1-2}\cmidrule{4-6}\cmidrule{8-10}\cmidrule{12-14}\cmidrule{16-18}          &       &       & All-zero &       & 100\% &       & rMCT  &       & 77.5\% &       & EMCDT &       & 75.9\% &       & CT    &       & 77.3\% \\
    Direct & \multirow{2}[0]{*}{\textcolor[rgb]{ 1,  0,  0}{73.1\%}} & \textcolor[rgb]{ 1,  0,  0}{} & \textcolor[rgb]{ .651,  .651,  .651}{(as $\vz^0$)} & \hspace{-0.1in} SAB   & 97.6\% &       & \textcolor[rgb]{ .651,  .651,  .651}{(as $\vz^0$)} & \hspace{-0.1in} SAB   & \textcolor[rgb]{ 1,  0,  0}{73.1\%} &       & \textcolor[rgb]{ .651,  .651,  .651}{(as $\vz^0$)} & \hspace{-0.1in} SAB   & \textcolor[rgb]{ 1,  0,  0}{73.1\%} &       & \textcolor[rgb]{ .651,  .651,  .651}{(as $\vz^0$)} & \hspace{-0.1in} SAB   & \textcolor[rgb]{ 1,  0,  0}{73.1\%} \\
    Search &       & \textcolor[rgb]{ 1,  0,  0}{} & $\drsh$ & \hspace{-0.1in} DolMD & 97.6\% &       & $\drsh$ & \hspace{-0.1in} DolMD & \textcolor[rgb]{ 1,  0,  0}{73.1\%} &       & $\drsh$ & \hspace{-0.1in} DolMD & \textcolor[rgb]{ 1,  0,  0}{73.1\%} &       & $\drsh$ & \hspace{-0.1in}  DolMD & \textcolor[rgb]{ 1,  0,  0}{73.1\%} \\
          &       &       &       & \hspace{-0.1in} SilMD & 97.6\% &       &       & \hspace{-0.1in} SilMD & 73.3\% &       &       &  \hspace{-0.1in}SilMD & 73.3\% &       &       & \hspace{-0.1in} SilMD & 73.3\% \\
    \bottomrule
    \end{tabular}%
  \label{tab:4}%
  \end{subtable}
\end{table}%

Table \ref{tab:sctp-hearn} highlights the fact that the performance of local methods, such as the proposed algorithms and SAB, often varies greatly with the initial solution. In Setting (i), when starting from an all-zero initial point, all three local algorithms, SAB, DolMD, and SilMD, easily reached the global optimum.  However, they all struggled in Setting (ii) from that same initial point, being trapped in a local optimum that happens to be the global optimum for Setting (i).

Since the heuristics are capable of breaking the trap of local optima, we next experiment with the idea of feeding their solutions to the local algorithms as initial points (see Columns 3-5 in Table \ref{tab:sctp-hearn}). The results are quite interesting.  In Setting (i), the solution provided by rMCT has no discernible effect. Yet, when using the solution provided by EMCDT, SAB became trapped by that solution and was unable to reach the global optimum.  The same applies to the solution by CT; the only difference is that this time, both SAB and DolMD were stuck.  The impact of the heuristic-inspired initial solutions is much more benign in Setting (ii), however.  No matter which heuristic solution is fed to the local algorithms, the result is universally positive in that a smaller relative excessive delay can be achieved after a local search.  In fact, equipped with these enhanced initial points, both SAB and DolMD always reached global optima.

The observations above suggest that combining local algorithms with suitable heuristics could provide a powerful tool for finding high-quality SCTP solutions.

\subsubsection{Larger networks} For Sioux-Falls, Barcelona, and Chicago Sketch, the DA algorithm is employed to produce benchmark solutions, as reported in the second row of Table \ref{tab:sctp}.  The performance of the three heuristics is then compared against that of the proposed local algorithms (DolMD and SilMD); see rows 4 - 6 in the table.  In each instance, four different initial points --- the all-zero initial solution,  along with the solutions obtained by the three heuristics --- are fed to the local algorithms.   The equilibrium threshold value is always set to $10^{-6}$ for Sioux Falls and $10^{-4}$ for the other two networks. For SilMD, we set $T = 40$ for Sioux Falls and $20$ for the other two.  
\begin{table}[htbp]
  \centering
  \scriptsize
  \caption{Comparison of different algorithms in solving SCTP on the Sioux Falls, Barcelona, Chicago-Sketch instances. The second row reports the performance of the DA algorithm. The third row reports the results of SAB, DolMD, and SilMD with an all-zero initial solution. Rows 4 to 6 report the performance of DolMD and SilMD with an initial solution given by one of the three heuristics (rMCT, EMCDT, and CT) as well as the three heuristics themselves.}
  \vspace{0.05in}
    \begin{tabular}{cccccccrccccccccc}
    \toprule
    \multicolumn{2}{c}{\multirow{2}[4]{*}{Method}} &       & \multicolumn{4}{c}{Sioux Falls} &       & \multicolumn{4}{c}{Barcelona} &       & \multicolumn{4}{c}{Chicago Sketch} \\
\cmidrule{4-7}\cmidrule{9-12}\cmidrule{14-17}    \multicolumn{2}{c}{} &       & r.e.d. & time (s) & iters & tpi (s)   &       & r.e.d. & time (min) & iters & tpi (s)  &       & r.e.d. & time (min) & iters & tpi (s) \\
    \midrule
    \multicolumn{2}{c}{DA} &       & \textcolor[rgb]{ 1,  0,  0}{16.5\%} & 2046  & 20000 & 0.10  &       & {82.8\%} & 865   & 20000 & 2.6   &       & {101.8\%} & 1645  & 10000 & 9.9 \\
    \midrule
    All-zero &       &       & 100\% &       &       &       &       & 100\% &       &       &       &       & 100\% &       &       &  \\
    \multirow{2}[1]{*}{$\drsh$} & \hspace{-0.1in} DolMD &       & 21.3\% & 18    & 56    & 0.33  &       & 74.5\% & 4.2   & 183   & 1.4   &       & \textcolor[rgb]{ 1,  0,  0}{91.8\%} & 23    & 39    & 36  \\
          & \hspace{-0.1in}  SilMD &       & 23.7\% & 1.8   & 56    & 0.032  &       & 76.0\% & 0.7   & 30    & 1.3   &       & 92.3\% & 9.5   & 200   & 2.8  \\
    \midrule
    EMCDT &       &       & 62.8\% & 1.8   & 36    & 0.051  &       & 88.4\% & 1.0   & 58    & 1.0   &       & 156.8\% & 5.4   & 47    & 6.8  \\
    \multirow{2}[1]{*}{$\drsh$} & \hspace{-0.1in} 
 DolMD &       & \textcolor[rgb]{ 1,  0,  0}{16.5\%} & 65    & 37    & 1.8   &       & \textcolor[rgb]{ 1,  0,  0}{73.7\%} & 7.3   & 118   & 3.7   &       & 92.1\% & 21    & 104   & 12  \\
          & \hspace{-0.1in}  SilMD &       & 23.6\% & 2.1   & 63    & 0.033  &       & 75.5\% & 1.0   & 80    & 0.74  &       & 93.6\% & 14    & 280   & 3.1  \\
    \midrule
    CT    &       &       & 20.9\% & 19    & 18    & 1.1   &       & 83.6\% & 7.5   & 507   & 0.89  &       & 97.1\% & 16    & 187   & 5.3  \\
    \multirow{2}[1]{*}{$\drsh$} & \hspace{-0.1in} 
 DolMD &       & \textcolor[rgb]{ 1,  0,  0}{16.5\%} & 25    & 34    & 0.72  &       & 74.4\% & 6.5   & 166   & 2.4   &       & \textcolor[rgb]{ 1,  0,  0}{91.8\%} & 27    & 49    & 33  \\
          & \hspace{-0.1in}  SilMD &       & 17.7\% & 1.6   & 51    & 0.031  &       & 76.1\% & 0.5   & 20    & 1.5   &       & 92.5\% & 7.4   & 120   & 3.7  \\
    \midrule
    rMCT  &       &       & 21.3\% & 15    & 406   & 0.038  &       & 78.7\% & 2.7   & 126   & 1.3   &       & 93.3\% & 14    & 80    & 10  \\
    \multirow{2}[1]{*}{$\drsh$} & \hspace{-0.1in} 
 DolMD &       & \textcolor[rgb]{ 1,  0,  0}{16.5\%} & 22    & 34    & 0.64  &       & \textcolor[rgb]{ 1,  0,  0}{73.7\%} & 11    & 164   & 4.0   &       & \textcolor[rgb]{ 1,  0,  0}{91.8\%} & 17    & 45    & 23  \\
          & \hspace{-0.1in}  SilMD &       & 17.8\% & 1.7   & 37    & 0.045  &       & 75.3\% & 1.3   & 160   & 0.47  &       & 92.7\% & 7.4   & 140   & 3.2  \\
    \bottomrule
    \end{tabular}%
    \vspace{-0.05in}
  \label{tab:sctp}%
\end{table}%

\textbf{DA.} The DA algorithm performs poorly for SCTP. While it was able to find the best solution (corresponding to a relative excessive delay of 16.5\%) for Sioux-Falls, the algorithm struggled to do so for the two larger networks.  On Chicago-Sketch, it failed completely, unable to improve over the no-toll solution after 10,000 iterations and more than 24 hours of computation time.  { %
This difficulty arises likely because the price of anarchy in practical congestion games is often small, producing many local solutions with similar objective values. The presence of a large number of similar solutions might have made it more difficult for standard global search algorithms, such as DA, to escape from sub-optimal local solutions.}

\textbf{EMCDT, CT, and rMCT.} EMCDT is the fastest among the three heuristics but also the crudest. In none of the three test instances could EMCDT reach the vicinity of the best solution identified in the experiments.  At the cost of higher computation time, CT and rMCT both delivered significantly better solutions than EMCDT in all instances.  For Barcelona and Chicago Sketch, rMCT outperformed CT in terms of both solution quality and computation time, though the performance gap is modest. 

\textbf{DolMD and SilMD.} In most test instances, the proposed algorithms delivered solutions with a lower relative excessive delay than the best achieved by any of the three heuristics. The only exception is Sioux Falls, for which both DolMD and SilMD appeared to be trapped in a similar local optimum when starting from the all-zero initial point. As a result, the solutions found by them are slightly worse than those found by CT and rMCT.  When fed with the initial solution provided by the heuristics, however, DolMD reached the best solution in all but { two cases (Barcelona with the initial solution from CT and Chicago Sketch with the initial solution from EMCDT).}  It is likely that these best solutions are, in fact, global optima, though we have no way to verify such a conjecture. { SilMD did not reach the best solutions but still outperformed all three heuristics in nearly all cases. } SilMD is also faster than CT and rMCT. On Chicago-Sketch, it achieved a relative excessive delay of 92.3\% (compared to the best value of 91.8\%) after 9.5 minutes (from the all-zero initial point), compared to 97.1\% achieved by CT (16 minutes) and 93.3\% by rMCT (14 minutes).   Once again, this result highlights SilMD's ability to closely track the performance of DolMD at a much lower computation cost.

}

\section{Conclusions}
\label{sec:conclusion}

{ Differentiable programming enables numerical evaluation of the gradients of complex functions and mappings, once considered impossible or too costly to differentiate, through a computation technique known as automatic differentiation. Our study explores how this new approach can be married with traditional methodologies --- notably bilevel optimization and game theory --- to develop high-performance local search algorithms for large-scale Stackelberg congestion games (SCGs).} Below, we first summarize the contributions and findings before commenting on future research.

{Using imitative logit dynamic (ILD), a classic model in evolutionary game theory, we developed a differentiable program (DiP) formulation of congestion games. Our analysis highlights two key benefits of the new formulation: its provable convergence to Wardrop equilibrium (WE) in general congestion games under mild conditions and its computational superiority when unrolled with AD.  The convergence result for ILD established here requires weaker conditions than those previously reported in the literature, hence a contribution in its own right.  The new SCG algorithms developed in this study, the double-loop mirror descent (DolMD) algorithm and its single-loop variant (SilMD),  were built upon this theoretical foundation.

By explicitly unrolling ILD's path to WE, the DolMD algorithm avoids expensive and unstable matrix operations necessary to carry out implicit differentiation in SAB methods.  As confirmed by the numerical experiments, this strategy bestows considerable and consistent computational advantages. While SAB methods matched DolMD in solution quality for small problems, they fell far behind in larger contests, often unable to converge.  Like other local search algorithms, DolMD can be trapped in poor local solutions. However, when fed with proper initial solutions (which can be provided by simple heuristics), it reliably reaches the best solution achievable by any algorithms included in the experiments. 
SilMD was inspired by the realization that carrying the anticipation of the followers' best response to terms may be a self-imposed computational obstacle.  Instead, the leader may only look ahead along the followers' evolution path for a few steps while updating their decisions in sync with the followers' decisions through a co-evolution process.  
The results from numerical experiments indicate that SilMD closely tracks DolMD in solution quality while requiring substantially less computation time. Compared to existing heuristics, which tend to do well only in certain applications under certain settings, SilMD is capable of providing solutions of decent quality with greater regularity and consistency. %

While both algorithms are general and scalable SCG solvers, they serve distinct purposes.  DolMD is the preferred choice for small problems or when the quality of the solution is the primary concern, whereas SilMD is more suitable when a reasonably good solution needs to be reached quickly. For either algorithm, initializing with solutions obtained by popular heuristics can significantly improve the chance of reaching high-quality local solutions.  

}

This study can be extended in many directions.  Many real-world congestion games do not admit a unique solution, even in the space of link flow.   An example that has gained prominence lately is the problem of managing mixed human-driven and autonomous vehicular traffic, which may introduce asymmetric link interactions. The lack of uniqueness in the follower's best response challenges the notion of anticipation: if the equilibrium is not even unique, how could the leader anticipate it? Is it reasonable to assume the leader can influence equilibrium realization? If so, how? 

The current framework, especially the concept of limited anticipation and co-evolution, hints at the possibility of further weakening or even abandoning the ``equilibrium-centric" point of view, at least in many transportation applications, where the lower-level problem is defined not so much by equilibrium conventionally construed as by continuous evolution of travel behavior.  Our framework would fit a behavior-centric approach perfectly since the ILD can be easily replaced with another evolutionary dynamics that embeds any travel behaviors the modeler considers essential in a given application context. In this regard, the flexibility offered by our framework seems unlimited. The challenge still resides with the issue of convergence.  Would a given evolutionary dynamics converge to a stationary point?  How should we interpret such stationary points? Could we use them as a reference point to evaluate alternative designs in the same way transportation planners had been using WE for such purposes?  These are the questions that future studies need to answer.

\section*{Acknowledgements}

This research is funded by the US National Science Foundation's Civil Infrastructure System (CIS) Program under the award number  CMMI \#2225087 and the Energy, Power, Control, and Networks (EPCN) Program under the award number ECCS \#2048075.

\bibliographystyle{apalike}
\begin{small}
\bibliography{example_paper}
\end{small}

\newpage
\appendix
\numberwithin{algorithm}{section}
\numberwithin{equation}{section}
\numberwithin{figure}{section}
\numberwithin{table}{section}

\section{Proofs in Section \ref{sec:day-to-day} and Further Discussions}
\label{app:proof}

\medskip
\subsection{Proof of Proposition \ref{prop:check-cocoercive}}
\label{app:check-cocoercive}

The following lemma \citep{marcotte1995convergence}provides conditions for checking cocoercivity.
\begin{lemma}
\label{lm:cocoercive}
Given a twice continuously differentiable and monotone function $f: \gY \subseteq \sR^n \to \sR^n$, suppose that $f(\vy)$ is $L$-Lipschitz continuous on $\gA$. If the matrix $\nabla f(\vy)^2 + (\nabla f(\vy)^2)^{\T}$ is positively semi-definite for all $\vy \in \gY$, then $f(\vy)$ is $1/4L$-cocoercive on $\gA$, i.e.,
\begin{equation}
    \langle f(\vy) - f(\vy'), \vy - \vy' \rangle \geq 1/4L \cdot \|f(\vy) - f(\vy') \|_2^2, \quad \text{for all}~\vy, \vy' \in \gY.
\end{equation}
\end{lemma}

We are now ready to prove Proposition \ref{prop:check-cocoercive}. First, as $\gX$ is a compact set, the twice continuously differentiability of $u(\cdot; \vz)$ directly implies that the function is $H_{\vz}$-Lipschitz continuous ($H_{\vz} = \max_{\vx \in \gX} \|u(\vx; \vz)\|_2$) on $\gX$. Hence, according to Lemma \ref{lm:cocoercive}, we claim that $u(\cdot; \vz)$ is $1/4H_{\vz}$-cocoercive on $\gX$. Thus, given any two $\vp, \vp' \in \gP$, setting $\vx = \bar \mLambda \vp$ and $\vx' = \bar \mLambda \vp'$ and letting $d_{\max} = \max_{w \in \gW} \{\evd_w\}$, we then have
\begin{equation}
\begin{split}
    &\langle c(\vp; \vz) - c(\vp'; \vz), \vp - \vp' \rangle \geq \frac{1}{d_{\max}} \cdot
    \langle u(\vx; \vz) - u(\vx'; \vz), \vx - \vx' \rangle \\
    &\qquad \geq \frac{1}{4d_{\max} \cdot H_{\vz}} \cdot \|u(\vx; \vz) - u(\vx'; \vz)\|_2^2 \geq \frac{1}{4d_{\max} \cdot H_{\vz} \cdot \|\mLambda\|_2^2} \cdot \|c(\vp; \vz) - c(\vp'; \vz)\|_2^2.
\end{split}
\end{equation}
Hence,  $c(\vp)$ is  $1/4L_{\vz}$-cocoercive if we set $L_{\vz} = d_{\max} \cdot H_{\vz} \cdot \|\mLambda\|_2^2$. \hfill \QED

\subsection{Proof of Theorem \ref{thm:convergence}}
\label{app:main-proof}

 Based on Proposition \ref{prop:ed-ild}, Lemma \ref{lm:kl-finite}, and a well-known equivalence between a convex program and a VIP \citep{kinderlehrer2000introduction}, we first establish the following lemma.
\begin{lemma}
\label{lm:vi-md}
For any $\vp^t \in \gP$, $\vp^{t + 1} = h(\vp^t; \vz)$ if and only if for all $w \in \gW$,
\begin{equation}
    \langle r \cdot c_w(\vp^t; \vz) + \nabla \phi_w(\vp_w^{t + 1}) - \nabla \phi_w(\vp_w^t), \vp_w - \vp_w^{t + 1}\rangle  \geq 0, \quad \forall \vp_w \in \gQ_w(\vp_w^t).
    \label{eq:fixed-proof}
\end{equation}
\end{lemma}
\begin{proof}
    We refer the readers to Appendix \ref{app:vi-md} for the proof.
\end{proof}
Based on Lemma \ref{lm:vi-md}, we then derive the following lemma, which provides conditions that guarantee a fixed point of ILD to be a WE. Plainly, it indicates that under the assumption that $c(\cdot; \vz)$ is $c_{\vz}$-cocoercive, a fixed point $\hat \vp$ is a WE as long as one can find a WE strategy $\vp^*$ whose support contains $\hat \vp$'s support.  
\begin{lemma}
\label{prop:fixed-point}
{ Suppose that $c(\cdot; \vz)$ is $c_{\vz}$-cocoercive on $\gP$. Given any $\hat \vp \in \gP$ such that $\hat \vp = h(\hat \vp; \vz)$, if  there exits $\vp^* \in \gP^*(\vz)$ such that $\vp^* \in \gQ(\hat \vp)$, then $\hat \vp \in \gP^*(\vz)$.}
\end{lemma}
\begin{proof}
    We refer the readers to Appendix \ref{app:fixed-point} for the proof.
\end{proof}

Then, to analyze the convergence of $\vp^t$, for any $\vp, \vp' \in \gP$, we define $\tilde D_{\phi}(\vp, \vp') = \sum_{w \in \gW} D_{\phi_w}(\vp_w, \vp_w')$. Based on Lemma \ref{lm:vi-md} and the $c_z$-cocoercivity of $c(\cdot; \vz)$, we then prove the following lemma.

\begin{lemma}
    \label{lm:final-lm}
    Suppose that the assumptions in Theorem \ref{thm:convergence} hold. For \textit{all} $\vp^* \in \gP^*(\vz)$, we have
\begin{equation}
     \tilde D_{\phi}(\vp^*, \vp^t) - \tilde D_{\phi}(\vp^*, \vp^{t + 1}) \geq \frac{2 c_{\vz} - r}{4 c_{\vz}} \cdot \|\vp^t - \vp^{t + 1} \|_2^2.
     \label{eq:main-mid}
\end{equation}
\end{lemma}
\begin{proof}
    See Appendix \ref{app:convergence} for the proof.
\end{proof}

Per assumption, there exists $\vp^* \in \gQ(\vp^0)$. For such $\vp^*$, Lemma \ref{lm:kl-finite} implies $\tilde D_{\phi}(\vp^*, \vp^{0}) < \infty$, and Lemma \ref{lm:final-lm} implies the divergence $\tilde D_{\phi}(\vp^*, \vp^{t})$ is monotonically decreasing as long as $r < 2c_{\vz}$.  Thus, the limit of $D_{\phi}(\vp^*, \vp^t)$ exists according to the monotone convergence theorem. Hence, by letting $t\to \infty$ on both sides of Equation \eqref{eq:monotone-proof3}, the squeeze theorem ensures $\| \vp^t - \vp^{t + 1}\|_2 \to 0$.

We should note the existence of $\lim_{t \to \infty} \tilde D_{\phi}(\vp^*, \vp^{t})$ for some $\vp^*$ does not imply the limit must be zero. Nor does it ensure the convergence of $\{\vp_t\}$. 
However, since $\gP$ is compact, the Bolzano-Weierstrass theorem  guarantees the sequence $\{\vp^t\}$ must have a convergent subsequence $\{\vp^{t_j}\}$.
Denote $\hat{\vp}$ as the limit of $\vp^{t_j}$ when $j \to \infty$. Assume  $\|h(\hat \vp; \vz) - \hat \vp\| = \delta$ for some $\delta > 0$.
That is, the limit of the subsequence $\hat \vp$ is not a fixed point of the ILD. 
We proceed to establish a contradiction. Note that
\begin{equation}
\begin{split}
    &\| \vp^{t_j + 1} -  \vp^{t_j}\|_2 =\|h(\vp^{t_j}; \vz) -  \vp^{t_j}\|_2 = \|h(\vp^{t_j}; \vz) - h(\hat{\vp}; \vz) + h(\hat{\vp}; \vz) - \hat{\vp} + \hat{\vp} -  \vp^{t_j}\|_2  \\
	&\quad \geq \|h(\hat{\vp}; \vz) - \hat{\vp}_i\|_2 - \|h(\vp^{t_j}; \vz) - h(\hat{\vp}; \vz) + \hat{\vp} -  \vp^{t_j} \|_2 \geq \delta - \|h(\vp^{t_j}; \vz) - h(\hat{\vp}; \vz) \|_2 - \|\hat{\vp} -  \vp^{t_j} \|_2.
\end{split}
\label{eq:bound-h}
\end{equation}
Letting $j \to \infty$ on both sides of Equation \eqref{eq:bound-h}, the left-hand side then converges to 0. Meanwhile, on the right-hand side, we also have $\|\hat{\vp} - \vp^{t_j} \| \to 0$ and thus $\|h(\vp^{t_j}; \vz) - h(\hat{\vp}; \vz)\| \to 0$ due to the continuity of $h(\vp; \vz)$ with respect to $\vp$. We then have
$0 = \|h(\vp^{t_j}; \vz) -  \vp^{t_j}\|_2
	\geq \delta > 0$, a contradiction. Thus,  $\hat{\vp} = h(\hat{\vp}; \vz)$.

We next claim  $\vp^* \in \gQ(\vp^0)$ implies  $\vp^* \in \gQ(\hat \vp)$.
Otherwise, $D_{\phi}(\vp^*, \hat \vp)$ would be unbounded, which is impossible given  $D_{\phi}(\vp^*, \vp^t)$ is  monotonically decreasing. This, along with the fact that $\hat \vp$ is a fixed point of ILD,  allows us to invoke  Lemma \ref{prop:fixed-point} to show the limit  $\hat \vp$ is indeed a WE. \hfill \QED

\subsubsection{Proof of Lemma \ref{lm:vi-md} }
\label{app:vi-md}

The following lemma \citep{kinderlehrer2000introduction} characterizes the relation between a convex program and a VIP. 
\begin{lemma}
\label{lm:optimization}
Given a closed and convex set $\gY \subseteq \sR^m$ and a continuously differentiable and convex function $f: \gY \to \sR$, consider the optimization problem $\min_{\vy \in \gY} f(\vy)$. Then $\vy^* \in \gY$ is an optimal solution to this convex program if and only if
$\langle\nabla f(\vy^*), \vy - \vy^*\rangle \geq 0$ for all $\vy \in \gY$.
\end{lemma}

Based on Lemma \ref{lm:optimization}, we can then prove Lemma \ref{lm:vi-md}. First, we have $\vp^{t + 1} = h(\vp^t; \vz)$ if and only if $\vp_w^{t + 1}$ solves Equation \eqref{eq:mirror-descent-w} per Proposition \ref{prop:ed-ild}. As the objective function in Equation \eqref{eq:mirror-descent-w}, denoted as $f_w(\vp_w) = r \cdot \langle c_w(\vp^t; \vz), \vp_w \rangle + D_{\phi_w}(\vp_w, \vp_w^t)$, is finite if and only if $\vp_w \in \gQ_w(\vp_w^t)$, we may equivalently minimize $f_w(\vp_w)$ over $\gQ_w(\vp_w^t)$, which is a closed set. Meanwhile, we have $\nabla f_w(\vp_w) = r \cdot c_w(\vp^t; \vz) + \nabla \phi_w(\vp_w) - \nabla \phi_w(\vp_w^t)$ and $\nabla^2 f_w(\vp_w) = \nabla^2 \phi_w(\vp_w)$, which is positively definite as $\phi_w$ is strongly convex on $\gP_w$. Hence, $f_w(\vp_w)$ is also strongly convex on $\gP_w$.  The proof is then concluded by directly applying Lemma \ref{lm:optimization}.
\hfill \QED

\subsubsection{Proof of Lemma \ref{prop:fixed-point}}
\label{app:fixed-point}

    {
    
    By replacing both $\vp^{t}$ and $\vp^{t + 1}$ by $\hat \vp$ in Equation \eqref{eq:fixed-proof} in Lemma \ref{lm:vi-md} and summarizing it over $w \in \gW$, we arrive at the following: $\hat \vp \in \gP$ satisfies $\hat \vp = h(\hat \vp; \vz)$ if and only if \begin{equation}
        \langle c(\hat \vp; \vz), \vp - \hat \vp \rangle  \geq 0, \quad \forall \vp \in \gQ(\hat \vp).
    \end{equation}
    As $\vp^* \in \gQ^*(\hat \vp)$ by assumption, we then have $\langle c(\hat \vp; \vz), \vp^* - \hat \vp \rangle  \geq 0$, which subsequently implies that 
    \begin{equation}
        \langle c(\hat \vp; \vz) - c(\vp^*; \vz), \hat \vp - \vp^*\rangle + \langle c(\vp^*; \vz), \hat \vp - \vp^*\rangle = - \langle c(\hat \vp; \vz), \vp^* - \hat \vp \rangle \leq 0.
    \end{equation}
    Noting that $\langle c(\hat \vp; \vz) - c(\vp^*; \vz), \hat \vp - \vp^*\rangle \geq c_{\vz} \cdot \|c(\hat \vp; \vz) - c(\vp^*; \vz)\|_2^2$ by the $c_{\vz}$-cocoercivity of $c(\cdot; \vz)$, we subsequently obtain
    \begin{equation}
        c_{\vz} \cdot \|c(\hat \vp; \vz) - c(\vp^*; \vz)\|_2^2 + \langle c(\vp^*; \vz), \hat \vp - \vp^*\rangle \leq 0.
    \end{equation}
    As $\langle c(\vp^*; \vz), \hat \vp - \vp^*\rangle \geq 0$ by Proposition \ref{prop:vi-formulation}, we must have $c(\hat \vp; \vz) = c(\vp^*; \vz)$ as well as 
    \begin{equation}
        \langle c(\vp^*; \vz), \hat \vp - \vp^*\rangle = 0.
        \label{eq:contra}
    \end{equation}
    Denoting $b_w^* = \min_{k' \in \gK_w} \{c_k(\hat \vp)\}$, we then prove $\evc_k^* > b_w^* \Rightarrow \hat \evp_k = 0$  for all $w \in \gW$ and $k \in \gK_w$, which shall indicate $\hat \vp$ is a WE by definition. To prove this, suppose that there exist $w' \in \gW$ and $k' \in \gK_{w'}$ such that $c^*_{k'} > b_{w'}^*$ and $\evp_{k'} > 0$. Denoting $\delta = c^*_{k'} - b_{w'}^*$, we will then have
    $
        \langle c(\hat \vp; \vz), \hat \vp \rangle  \geq \sum_{w \in \gW} b_w^* + \delta \cdot \evp_{k'}.
    $
    Noting that $\langle \vc(\vp^*; \vz), \vp^*\rangle = \sum_{w \in \gW} b_w^*$ and $\vc(\vp^*; \vz) = \vc(\hat \vp; \vz)$, we then derive that $\langle c(\vp^*; \vz), \hat \vp - \vp^* \rangle > 0$, which contradicts Equation \eqref{eq:contra}.  Hence, we must have $\hat \vp \in \gP^*(\vz)$, which concludes the proof. \hfill \QED
}

\subsubsection{Proof of Lemma \ref{lm:final-lm}}
\label{app:convergence}

We first provide several lemmas that will be used in the proof. In all of these lemmas, we assume that $\gA$ is a closed convex set. 
\begin{lemma}[\citet{chen1993convergence}]
\label{lm:triangle}
For any Bregman divergence $D_{\phi}: \gY \times \gY \to \sR \cup \{\infty\}$ and any $\vy, \vy', \vy'' \in \gY$, we have
\begin{equation}
	D_{\phi}(\vy, \vy') + D_{\phi}(\vy', \vy'') - D_{\phi}(\vy, \vy'') = \langle \vy - \vy', \nabla \phi(\vy'') - \nabla \phi(\vy') \rangle.
	\label{eq:triangle}
\end{equation}
\end{lemma}

\begin{lemma}
\label{lm:bregman-ieq}
If $\phi$ is $\sigma$-strongly convex with respect to $\|\cdot\|$ on $\gY$, then for all $\vy, \vy' \in \gY$, the induced Bregman divergence $D_{\phi}(\vy, \vy')$ satisfies
\begin{equation}
	D_{\phi}(\vy, \vy') \geq \frac{\sigma}{2} \cdot  \|\vy - \vy'\|^2.
	\label{eq:bregman-bound}
\end{equation}
\end{lemma}

\begin{lemma}[\citet{beck2003mirror}]
\label{lm:entropy-norn}
The negative entropy function $\phi(\vy) = \langle \vy, \log(\vy) \rangle$, which  induces the KL divergence, is $1$-strongly convex with respect to  $\ell_1$ norm $\|\cdot\|_1$. 
\end{lemma}

\begin{lemma}
\label{lm:l1-l2}
For all $\vy \in \gY$, we have $\|\vy\|_2 \leq \|\vy\|_1$.
\end{lemma}
\begin{lemma}[\citet{marcotte1995convergence}]
\label{lm:bound}
For any two vectors $\vy, \vy' \in \gY$, we have
\begin{equation}
	\langle\vy, \vy'\rangle - \|\vy'\|_2^2 \leq \frac{1}{4} \cdot \|\vy\|_2^2.
\end{equation}
\end{lemma}

Now we are ready to prove Theorem \ref{thm:convergence}.
For all $\vp^* \in \gP^*(\vz)$, Lemma \ref{lm:triangle} first implies that 
\begin{equation}
    D_{\phi_w}(\vp_w^*, \vp_w^{t + 1}) = D_{\phi_w}(\vp_w^*, \vp_w^t) - D_{\phi_w}(\vp_w^{t + 1},  \vp_w^t) + \langle \nabla \phi(\vp_w^{t}) - \nabla \phi(\vp_w^{t + 1}), \vp_w^* - \vp_w^{t + 1}\rangle
\label{eq:two-lemmas}
\end{equation}
Combining Equation \eqref{eq:fixed-proof} in Lemma \ref{lm:vi-md} and Equation \eqref{eq:two-lemmas}, we arrive at 
\begin{equation}
\begin{split}
    D_{\phi_w}( \vp_w^*, \vp_w^t) - D_{\phi_w}(\vp_w^*, \vp_w^{t + 1}) &\geq D_{\phi_w}(\vp_w^{t + 1}, \vp_w^t) + r \cdot \langle c_w(\vp^t; \vz), \vp_w^{t + 1} - \vp_w^*\rangle.
    \label{eq:middle-0}
\end{split}
\end{equation}
Summarizing Equation \eqref{eq:middle-0} for all $w \in \gW$ leads to
\begin{equation}
\begin{split}
    \tilde D_{\phi}( \vp^*, \vp^t) - \tilde D_{\phi}(\vp^*, \vp^{t + 1}) &\geq \tilde D_{\phi}(\vp^t, \vp^{t + 1}) + r \cdot \langle c(\vp^t; \vz), \vp^{t + 1} - \vp^*\rangle \\
    &\geq \tilde D_{\phi}(\vp^t, \vp^{t + 1}) + r \cdot \langle c(\vp^t; \vz) - c(\vp^*; \vz), \vp^{t + 1} - \vp^*\rangle,
    \label{eq:middle-i}
\end{split}
\end{equation}
where the second inequality holds because $\langle c(\vp^*; \vz), \vp^{t + 1} - \vp^* \rangle \geq 0$.
Below we further bound the two terms on the right-hand side of Equation \eqref{eq:middle-i}. First, 
\begin{equation}
    \tilde D_{\phi}(\vp^{t + 1}, \vp^t) \geq \frac{1}{2} \cdot \sum_{w \in \gW} \|\vp_w^t - \vp_w^{t + 1}\|_1^2 \geq \sum_{w \in \gW} \|\vp_w^t - \vp_w^{t + 1}\|_2^2 =  \frac{1}{2} \cdot \|\vp^t - \vp^{t + 1} \|_2^2,
    \label{eq:term-1}
\end{equation}
where the first inequality follows from Lemmas \ref{lm:bregman-ieq} and \ref{lm:entropy-norn}, and the second follows from Lemma \ref{lm:l1-l2}. Second, 
\begin{equation}
\begin{split}
    &\langle c(\vp^t; \vz) - c(\vp^*; \vz), \vp^{t + 1} - \vp^*\rangle =  \langle c(\vp^t; \vz) - c(\vp^*; \vz), \vp^{t + 1} - \vp^t + \vp^t - \vp^*\rangle \\
    &\qquad \geq \langle c(\vp^t; \vz) - c(\vp^*; \vz), \vp^{t + 1} - \vp^t \rangle + c_z \cdot \|c(\vp^t; \vz) - c(\vp^*; \vz)\|_2^2 \geq -\frac{1}{4c_{\vz}} \cdot \|\vp^t - \vp^{t + 1} \|_2^2,
    \label{eq:term-2}
\end{split}
\end{equation}
where the first inequality is guaranteed by  the cocoercivity of $c(\vp; \vz)$ and the second follows from Lemma \ref{lm:bound}. Plugging  Equations \eqref{eq:term-1} and \eqref{eq:term-2} into Equation \eqref{eq:middle-i} results in
\begin{equation}
     D_{\phi}(\vp^*, \vp^t) - D_{\phi}(\vp^*, \vp^{t + 1}) \geq \frac{2 c_{\vz} - r}{4 c_{\vz}} \cdot \|\vp^t - \vp^{t + 1} \|_2^2,
     \label{eq:monotone-proof3}
\end{equation}
which concludes the proof. \hfill \QED

\subsection{Theorem \ref{thm:convergence} and existing results}
\label{app:explanation}

{

{

Although Theorem \ref{thm:convergence} is new, its relation with existing results in literature warrants some remarks. We first note that our convergence analysis applies to congestion games both {with} and {without} a ``potential function,"  i.e., a function whose gradient with respect to route flows always equals route cost \citep{beckmann1956studies}. It is well known the existence of such a function hinges on the symmetry of the Jacobian matrix of the link cost function with respect to link flow \citep{sheffi1985urban}.

For potential congestion games, it is straightforward to verify that the convergence of ILD is equivalent to the convergence of minimizing the potential function by MD. %
To ensure the convergence of MD for solving a convex program, the following results are standard: (i) for convex objective functions, a progressively decreasing step size is necessary; (ii) for strictly or strongly convex objective functions,  a small constant step size is sufficient \citep[see, e.g.,][]{beck2003mirror,krichene2015convergence, doan2018convergence, radhakrishnan2020linear}. However, the potential function in congestion games, even when it exists, is typically convex but not strictly convex with respect to route flow \citep{sheffi1985urban}. Hence, our finding, which asserts the convergence of MD with a constant step size and cocoerciviity (weaker than strong convexity) is new.

For congestion games without a potential function, our result is related to \citet{mertikopoulos2019learning}. The difference is twofold. First, their convergence requires the equilibrium set to satisfy a special ``variational stability" condition, which extends \citet{smith1982evolution}'s notion of ``evolutionary stability." However, there are no readily available results that can check whether the set of WE strategies in a congestion game is variationally stable. In contrast, our result builds on the cocoercivity of the route cost function, which can be easily verified  (see Proposition \ref{prop:check-cocoercive}). Second, \citet{mertikopoulos2019learning} also require the step size of MD to progressively decrease at an appropriate rate.

Our analysis was also inspired by \citet{marcotte1995convergence}, who established the convergence of the Euclidean projection method in congestion games under similar conditions. 
ILD and the Euclidean projection method share a similar algorithmic structure. 
However, it is not straightforward to extend \citet{marcotte1995convergence}'s analysis to ILD because a fixed point of the Euclidean projection method must be a WE, while a fixed point of ILD may not. This crucial difference means that, in the convergence proof, we must check whether the limiting point of ILD is indeed a WE. 
}

}

\section{Multi-Class Extension}
\label{app:extension}
The differentiable bilevel programming approach can be easily extended to handle more general settings of SCGs.  This advantage is highlighted here with a classical extension of the congestion game: the multi-class congestion game \citep{harker1988multiple}.
In a multi-class congestion game, travelers are divided into different classes according to certain characteristics.   
Denote $\gM$ as the set of classes and let $\vd_m$  represent the number of travelers in class $m \in \gM$. Similarly, we define $\vq_m = \mSigma^{\T} \vd_m$.   The route choice of travelers in class $m$ is represented by a vector $\vp_m \in \gP = \{\vp \geq \vzero: \mSigma \vp = \vone\}$. Accordingly,  route flow $\vf_m = \diag(\vq_m) \vp_m$ and  link flow $\vx_m = \mSigma \vf_m$. We then denote the link cost experienced by travelers in class $m$ as $\vu_m = u_m(\vx_{}; \vz)$, where $\vx_{} = (\vx_{m})_{m \in \gM}$, and the route cost $\vc_m = \mLambda^{\T} \vu_m$.  Denoting $\vp = (\vp_m)_{m \in \gM} \in \gP^{|\gM|}$ as the joint route choice of all travelers,  $c_m: \gP^{|\gM|} \to \sR^{|\gK|}$ can be defined as $c_m(\vp; \vz) = \mLambda u_m(\vx; \vz)$,  where $\vx = (\vx_m)_{m \in \gM} = (\bar \mLambda \vp_m)_{m \in \gM}$, where $\bar \mLambda = \mLambda \diag(\vq_m)$.

The above setting allows (i) travelers from different classes to experience a different cost on the same link, and (ii) the cost on a link is affected by flows on other links in an asymmetric manner. %

\begin{example}[Mixed autonomy]
\label{eg:mixed-autonomy}
Assume travelers drive two different types of vehicles: ($|\gM| = 2$): connected and autonomous vehicles (CAVs, $m = 1$) and human-driven vehicles (HDV, $m = 2$), and the effective capacity on a given link vary with the share of CAVs using it. 
\cite{bahrami2020optimal}, for example, suggested the link cost be computed using the following revised  Bureau of Public Road (BPR) function:
\begin{equation}
\begin{split}
u_{\text{time},a}(\evx_{1,a}, \evx_{2, a}) &= \evu_{0, a} \cdot \left(1 + 0.15 \cdot \left(\frac{\evx_{\text{sum},a}}{(1 + \eta \cdot ( \evx_{1, a} / \evx_{\text{sum}, a})^2) \cdot \evv_{0,a}}\right)^4 \right),
\end{split}
\label{eq:revised-bpr}
\end{equation}
where $\evx_{\text{sum},a} = \evx_{1,a} +  \evx_{2, a}$ is the sum of CAV and HDV flows on link $a$,  $\evu_{0, a}$ is the free-flow travel time on link $a$, $\evv_{0,a}$ is the capacity when all vehicles on link $a$ are HDVs, and $\eta$ is a parameter that measures the increase in capacity after all HDVs are replaced by CAVs.  Suppose a class-specific toll is levied on link $a$. That is, CAV and HDV drivers must each pay a different toll to use link $a$, denoted as $\evbeta_{1, a}$ and $\evbeta_{2, a}$, respectively. Thus, the class-specific link cost function reads   $u_{m}(\vx_1, \vx_2) = u_{\text{time}}(\vx_1, \vx_2) + \gamma \cdot \vbeta_m$ ($m = 1, 2$), where $\gamma$ is the time value of money. 
\end{example}

Under the multi-class setting, a WE strategy can be similarly defined as follows.
\begin{definition}
A joint route choice $\vp^* \in \gP^{|\gM|}$ is a WE strategy of a multi-class congestion game if  $c_{m, k}(\vp^*; \vz) > b_{m, w}^* \Rightarrow \evp_{m, k}^{*} = 0$, where $b_{m, w}^* = \min_{k' \in \gK_w} c_{m, k'}(\vp^*; \vz)$, for all $m \in \gM$, $w \in \gW$, and $k \in \gK_w$,
\end{definition}

The VIP formulation can be directly extended to handle the multi-class problem  \citep{dafermos1980traffic}.
\begin{proposition}
\label{prop:mue-vi}
Letting $c(\vp_{}; \vz) = (c_m(\vp_{}; \vz))_{m \in \gM}$, then a joint route choice $\vp^* \in \gP^{|\gM|}$ is a WE strategy if and only if
\begin{equation}
    \langle c(\vp^*; \vz), \vp - \vp^{*} \rangle = \sum_{m \in \gM} \langle c_m(\vp^*; \vz), \vp_m - \vp_m^{*} \rangle \geq \vzero, \quad \forall \vp \in \gP^{|\gM|}.
\end{equation}
\end{proposition}

We next discuss how the proposed framework can be adapted to solve an SCG with a multi-class congestion game at its lower level.  The key is to reformulate the multi-class congestion game as a DiP.  Since the ILD operates on the zeroth-order (i.e., it needs no more information than  travel cost), the extension is straightforward, requiring only to redefine the $h$ function as  $h: \gP^{|\gM|} \times \gP \to \gP^{|\gM|}$, where
\begin{equation}
    h_{m, k}(\vp; \vz) = \frac{\evp_{m, k} \cdot \exp(-r \cdot c_{m, k}(\vp; \vz))}{\displaystyle  \sum_{k' \in \gK_w} \evp_{m, k'} \cdot \exp(-r \cdot c_{m, k'}(\vp; \vz))}, \quad \forall k \in \gK_w, \quad \forall w \in \gW, \quad \forall m \in \gM.
    \label{eq:def-h-mce}
\end{equation}
Would the dynamics defined by $\vp^{t + 1} = h(\vp^t; \vz)$ converge to a WE of the multi-class congestion game? Because WE and multi-class WE share the same form of VIP formulation (see Propositions \ref{prop:vi-formulation} and \ref{prop:mue-vi}),  we believe the answer is likely yes. That is, Theorem \ref{thm:convergence} may be extended to cover the multi-class case.  However, whether the function $c(\vp; \vz) = (c_m(\vp; \vz))_{m \in \gM}$ is cocoercive might be application specific.
We leave a thorough investigation of the convergence issues for future study.

With the new ILD, all that is left to do is reprogram $h(\vp; \vz)$ according to \eqref{eq:def-h-mce} in Algorithm \ref{alg:forward}. The rest will be taken care of by the framework itself, with little additional effort.

{

\section{Algorithms Included in the Numerical Experiments}
\label{app:comparing-alg}

We first provide additional implementation details of the proposed algorithms (Appendix \ref{app:DolMD}) before proceeding to introduce other general SCG algorithms included for comparison (Appendix \ref{app:other-general}).  Application-specific heuristics are presented in Appendix \ref{app:dedicated-heuristics}.

\subsection{DolMD and SilMD with route generation}
\label{app:DolMD}

}

\begin{algorithm}[ht]
   \caption{DolMD implemented with dynamic route generation.}
   \label{alg:stackelberg-rg}
\begin{algorithmic}[1]
{\footnotesize
   \State {\bfseries Input:} $\gK^+ \subseteq \gK$, $\vz^0 \in \gZ$,  step sizes $\rho > 0$ and $r > 0$, threshold values $\varepsilon > 0$ and $\xi > 0$.
   \vspace{1.5pt}
   \State Set $\vp^0 = (\evp_k^0)_{k \in \gK}$, where $\evp_k^0 = 1 / |\gK_w^+|$ if $k \in \gK_w^+$ and $\evp_k^0 = 0$ if $k \in \gK_w \setminus \gK_w^+$ for all $w \in \gW$.
    \vspace{1.5pt}
   \For{$i = 0, 1, \ldots$}
   \vspace{1.5pt}
    \State {\bfseries FP} (running Algorithm \ref{alg:forward} until convergence):
    \vspace{1.5pt}
    \For{$t = 0, 1, \ldots$}
        \State Run $\vp^{t + 1} = h(\vp^t; \vz^i)$. If $\delta(\vp^t; \vz^i) \leq \varepsilon$, break and set $T = t$ and $l^T = l(\bar \mLambda \vp^T; \vz^i)$.
        \vspace{1.5pt}
    \EndFor
    \vspace{1.5pt}
    \State {\bfseries BP}: Calculate $l_{\vz} = \partial l^T / \partial \vz_i$ by unrolling FP via AD (AD tools automatically programmed and executed Algorithm \ref{alg:backward}).
    \vspace{1.5pt}
    \State Set $\vz^{i + 1} = \argmin_{\vz \in \gZ}~\rho \cdot \langle l_{\vz}, \vz - \vz^{i} \rangle + D_{\psi}(\vz, \vz^i)$. 
    \vspace{1.5pt}
    \For{$w \in \gW$}
    \State Find $k^* = \argmin_{k \in \gK_w} c_k(\vp^T; \vz^i)$ through a shortest route algorithm and then set $\gK^+ = \gK^+ \cup \{k^*\}$. 
    \vspace{1.5pt}
    \EndFor
    \vspace{1.5pt}
    \State If $\gK^+$ is updated, reset $\vp^0 = (\evp_k^0)_{k \in \gK}$, where $\evp_k^0 = 1 / |\gK_w^+|$ if $k \in \gK_w^+$ and $\evp_k^0 = 0$ if $k \in \gK_w \setminus \gK_w^+$ for all $w \in \gW$.
    \vspace{1.5pt}
    \State If $\| \vz^{i} - \vz^{i - 1}\|_{\infty} < \xi$ and no routes are added to $\gK^+$, break and set $\vz^* = \vz^i$.
   \EndFor
}
\end{algorithmic}
\end{algorithm}

{ Algorithm \ref{alg:stackelberg-rg} details a version of DolMD in which routes are not enumerated beforehand but generated iteratively on the fly.  The initial $\vp^0$ is set according to the initial $\gK^+$: travelers are split equally to all routes currently present between each OD pair. At the end of each outer iteration, the shortest route between each OD pair is identified and included in $\gK^+$ if it is new. Whenever $\gK^+$ is updated, $\vp^0$ is reset accordingly.}

\begin{algorithm}[ht]
   \caption{SilMD implemented with dynamic route generation}
   \label{alg:cournot-rg}
\begin{algorithmic}[1]
{\footnotesize
   \State {\bfseries Input:} $\gK^+ \subseteq \gK$, $\vz^0 \in \gZ$, step sizes $\rho > 0$ and $r > 0$, threshold values $\varepsilon > 0$ and $\xi > 0$.
       \State Set $\vp^0 = (\evp_k^0)_{k \in \gK}$, where $\evp_k^0 = 1 / |\gK_w^+|$ if $k \in \gK_w^+$ and $\evp_k^0 = 0$ if $k \in \gK_w \setminus \gK_w^+$ for all $w \in \gW$.
    \vspace{1.5pt}
   \For{$i = 0, 1, \ldots$}
    
    \State {\bfseries FP}: Calculate $l^T = g^{(T)}(\vp^i; \vz^i)$.
    \vspace{1.5pt}
    \State {\bfseries BP}: Calculate $l_{\vz} = \partial l^T / \partial \vz^i$ by unrolling FP via AD.
    \vspace{1.5pt}
    \State Set $\vz^{i + 1} = \argmin_{\vz \in \gZ}~\rho \cdot \langle l_{\vz}, \vz - \vz^{i} \rangle + D_{\psi}(\vz, \vz^t)$ and $\vp^{i + 1} = h(\vp^i; \vz^i)$.
    \vspace{1.5pt}
    \For{$w \in \gW$}
    \vspace{1.5pt}
    \State Find $k^* = \argmin_{k \in \gK_w} c_k(\vp^i; \vz^i)$ through a shortest route algorithm. 
    \vspace{1.5pt}
    \State  If $k^* \notin \gK_w^+$, set $\gK^+ = \gK^+ \cup \{k^*\}$ and then initialize $\evp_k^i$ with a small positive value. 
    \vspace{1.5pt}
    \vspace{1.5pt}
    \EndFor
    \vspace{1.5pt}
    \State Set $\vp^i = \vp^i / \mSigma^{\T} \mSigma \vp^i$ (scaling $\vp^i$ so that it satisfies the constraint $\mSigma \vp^i = \vone$).
    \vspace{1.5pt}
    \State If $\delta(\vp^i; \vz^i) \leq \varepsilon$  and $\| \vz^{i} - \vz^{i - 1}\|_{\infty} < \xi$, break and set $\tilde \vz = \vz^i$.
   \EndFor
}
\end{algorithmic}
\end{algorithm}

{   %

Algorithm \ref{alg:cournot-rg} describes how SilMD is implemented with route generation. The main difference concerns how the route choice strategy is initialized when a new route is found. In SimMD, if a new route $i$ is added to the route set, its choice probability should be manually set.  Otherwise, the flow on that route will always remain zero since, unlike DolMD, route choice strategies are not reset every time a route expansion occurs (see Algorithm \ref{alg:stackelberg-rg}). 
 }

{

For both DolMD and SilMD, in our implementation, the parameter $r$ is fixed as a small constant, while $\rho$ is decreased at a harmonic rate. Meanwhile, noting that $\gZ$ is a box constraint in both CNDPs and SCTPs, we choose $D_{\psi}$ as the squared Euclidean distance when implementing both algorithms.

\subsection{Other general SCG algorithms}
\label{app:other-general}

\subsubsection{SAB algorithm} 
\label{app:sab}

Algorithm \ref{alg:sab} presents an SAB algorithm for solving SCGs, which shares a similar overall structure with DolMD. At each iteration, both algorithms also perform the following three tasks: (i) find the lower-level WE; (ii) evaluate the gradient of the leader's cost at WE; and (iii) improve the leader's decision based on the gradient. For Task (i), Algorithm \ref{alg:sab}  enjoys the benefit of a faster WE solver, namely iGP, rather than sticking to ILD as in DolMD. For Task 2, it relies on \citet{yang2007sensitivity}'s sensitivity analysis approach given in Proposition \ref{prop:diff-formula}. As for Task (iii), Algorithm \ref{alg:sab} also leverages the MD method, identical to DolMD.
}

\begin{algorithm}[ht]
   \caption{SAB algorithm for solving Problem \eqref{eq:bilevel}.}
   \label{alg:sab}
\begin{algorithmic}[1]
{\footnotesize
   \State {\bfseries Input:} $\vz^0 \in \gZ$, step size $\rho > 0$, threshold values $\varepsilon > 0$ and $\xi > 0$.
   \vspace{1.5pt}
   \For{$i = 0, 1, \ldots$}
   \vspace{1.5pt}
    \State Solve a $\vp^i \in \gP^*(\vz^i)$ via iGP, terminating which when $\delta(\vp^i; \vz^i) < \varepsilon$. Set $\vx^i = \bar \mLambda \vp^i$.
    \State Calculate $\nabla x^*(\vz^i)$ according to the sensitivity analysis method detailed in Proposition \ref{prop:diff-formula}.
    \vspace{1.5pt}
    \State Set $\vz^{i + 1} = \argmin_{\vz \in \gZ}~\rho \cdot \langle l_{\vz}, \vz - \vz^{i} \rangle + D_{\psi}(\vz, \vz^i)$, where $l_{\vz} =  \nabla_{\vz} l(\vx^i; \vz^i) + \nabla x^*(\vz^i) \cdot \nabla_{\vx} l(\vx^i; \vz^i)$.
    \vspace{1.5pt}
    \State If $\| \vz^{i} - \vz^{i - 1}\|_{\infty} < \xi$, break and set $\vz^* = \vz^i$.
   \EndFor
}
\end{algorithmic}
\end{algorithm}

{

As highlighted in Proposition \ref{prop:diff-formula}, the algorithm requires the selected WE strategy $\bar \vp^*$ satisfy the strict complementary condition, which enforces $\supp(\bar \vp^*) = \cup_{\vp^* \in \gP^*} \supp(\vp^*)$ (``no-route-left-behind").   Our experience suggests that the algorithm tends to work fine when we simply select $\bar \vp^*$ as a non-degenerate extreme point of $\gX^*$,  as suggested in \citet{tobin1988sensitivity}.  Note that this approach would avoid the cost of finding the MSIC for $[\mLambda_+^{\T}, \mSigma_+^{\T}]^{\T}$, though it may become infeasible when all extreme points of $\gX^*$ are degenerate \citep{josefsson2003pitfalls}.

\subsubsection{Cutting plane algorithm}
\label{app:cp}
Underlying the cutting plane (CP) algorithm is the following reformulation of the original problem \eqref{eq:bilevel}, 
\begin{equation}
\begin{split}
    \min_{\vz \in \gZ, \ \vp^* \in \gP}~~&l(\vx^*; \vz), \\
    \text{s.t.}~~&\vx^* = \bar \mLambda \vp^*, \quad \langle u(\vx^*; \vz), \vy - \vx^* \rangle \geq 0, \quad \forall \vy \in \gX,
\end{split} \label{eq:cutting-plane}
\end{equation}
where the lower-level WE is replaced by its link-based VIP formulation. Note that Problem \eqref{eq:cutting-plane} has an infinite number of constraints. \citet{lawphongpanich2004mpec} noted that the VIP in Problem \eqref{eq:cutting-plane} holds as long as
\begin{equation}
    \langle u(\vx^*; \vz), \vy^i - \vx^* \rangle \geq 0, \quad \forall~\text{vertex $\vy^i$ of $\gX$}.
    \label{eq:vertex-cut}
\end{equation}
This equivalence allows Problem \eqref{eq:cutting-plane} to be further transformed into a standard nonlinear program with a finite number of nonlinear constraints. To avoid enumerating all vertices of $\gX$, they further devised a cutting plane scheme that iteratively generates cuts (vertex $\vy^i$ of $\gX$). As described in Algorithm \ref{alg:cp}, eachsub-problem is equivalent to a shortest route problem: its solution $\vy^{i + 1}$ corresponds to the all-or-nothing assignment where link costs equal $u(\vx^i; \vz^i)$.
}

\begin{algorithm}[ht]
   \caption{Cutting plane algorithm for solving Problem \eqref{eq:bilevel}.}
   \label{alg:cp}
\begin{algorithmic}[1]
{\footnotesize
   \State {\bfseries Input:} a vertex $\vy^0$ of $\gX$.
   \vspace{1.5pt}
   \For{$i = 1, 2, \ldots$}
   \vspace{1.5pt}
    \State \textbf{The master problem}: let $(\vz^i, \vp^i)$ be the solution to
    \begin{equation*}
    \begin{split}
        \min_{\vz \in \gZ, \ \vp^* \in \gP}~~&l(\vx^*; \vz), \\
        \text{s.t.}~~&\vx^* = \bar \mLambda \vp^*, \quad \langle u(\vx^*; \vz), \vy^j - \vx^* \rangle \geq 0, \quad j = 0, \ldots, i - 1,
    \end{split} \label{eq:cutting-plane-finite-rg}
    \end{equation*}
    \label{ln:cp-master}
    \vspace{-0.15in}
    \State \textbf{The sub-problem}: solve $\vy^{i + 1} = \argmin_{\vy \in \gX} \langle u(\vx^{i}; \vz^{i}), \vy \rangle$, where $\vx^i = \bar \mLambda \vp^i$.
    \label{ln:cp-sub}
    \vspace{1.5pt}
    \State If $l^*(\vz^i) - l(\vx^i; \vp^i)$ is sufficiently small, break and set $\vz^* = \vz^i$.
   \EndFor
}
\end{algorithmic}
\end{algorithm}

{ When dealing with a toy problem like Braess, the CP algorithm can be implemented without the need to dynamically generate cuts since the vertices of  $\gX$ can be easily enumerated.  However, applying Algorithm \ref{alg:cp} directly to larger networks presents significant challenges for two main reasons. First, the master problem is a highly nonconvex and nonlinear program, making it difficult to solve globally and exactly. Second, the constraint $\vp^* \in \gP = \{\vp \in \sR_+^{|\gK|}: \mSigma \vp = \vone\}$ in the master problem requires route enumeration to create the OD-route incidence matrix $\mSigma$. %

\subsection{Application-specific heuristics}
\label{app:dedicated-heuristics}

\subsubsection{Heuristics for CNDPs}
\label{app:heuristics-cndp}
We implemented two heuristic methods: IOA and SO.}

\begin{algorithm}[ht]
   \caption{IOA algorithm for solving CNDPs.}
   \label{alg:ioa}
\begin{algorithmic}[1]
{\footnotesize
   \State {\bfseries Input:} $\vz^0 \in \gZ$, threshold values $\varepsilon > 0$ and $\xi > 0$.
   \vspace{1.5pt}
   \For{$i = 1, 2, \ldots$}
   \vspace{1.5pt}
    \State \textbf{The assignment step}: solve a $\vp^i \in \gP^*(\vz^{i - 1})$ via iGP, terminating which when $\delta(\vp^i; \vz^{i - 1}) < \varepsilon$; set $\vx^i = \bar \mLambda \vp^i$.
    \vspace{1.5pt}
    \State \textbf{The optimization step}: solve $\vz^{i + 1} = \argmin_{\vz \in \gZ} l(\vx^i; \vz)$.
    \vspace{1.5pt}
    \State If $\| \vz^{i} - \vz^{i - 1}\|_{\infty} < \xi$ converges, break and set $\tilde \vz = \vz^i$.
   \EndFor
}
\end{algorithmic}
\end{algorithm}

{
\noindent
\textbf{IOA.} As described in Algorithm \ref{alg:ioa}, our implementation solves the corresponding WE via the iGP algorithm. The optimization step minimizes $l(\vx^i; \vz)$ over $\vz \in \gZ$, a nonlinear program with simple box constraints that can be handled by standard nonlinear program solvers.}

\noindent
{ \textbf{SO}. The SO algorithm accepts $\tilde \vz$ from the following optimization problem  as an approximate solution
\begin{equation}
    (\hat \vz, \hat \vp) \in \argmin_{\vz \in \gZ, \ \vp \in \gP} l(\vx; \vz), \quad \text{s.t.}~\vx = \bar \mLambda \vp.
\label{eq:so}
\end{equation}
On large networks, we find solving \eqref{eq:so} with commercial convex program solvers is not efficient. This is because specifying the constraint $\vp^* \in \gP = \{\vp \in \sR_+^{|\gK|}: \mSigma \vp = \vone\}$ requires route enumeration. To improve efficiency, we employed a coordinate descent algorithm (see Algorithm \ref{alg:so}), which iterates between two steps.
\begin{itemize}
    \item \textbf{The assignment step}: given $\vz^{i - 1} \in \gZ$, find $\vp^i \in \gP$ that minimizes $l(\vx; \vz^{i - 1})$, where $\vx = \bar \mLambda \vp^i$. To see how this sub-problem can be transformed into a standard traffic assignment problem, we note that
    $
        l(\vx; \vz^{i - 1}) = \langle \vx, u(\vx; \vz^i) \rangle + \beta \cdot m(\vz^i),
    $
    which indicates that the sub-problem is equivalent to an SO traffic assignment problem (SO-TAP) that aims to find $\vp \in \gP$ to minimize the total travel time $\langle \vx, u(\vx; \vz^i) \rangle$. The resulting SO-TAP may be further transformed into an equivalent WE traffic assignment problem (WE-TAP) by adding marginal costs to the link cost functions \citep{sheffi1985urban},  which can be readily and efficiently solved via the iGP algorithm. Let us denote the solution set to this WE-TAP as $\gP_{\text{so}}^*(\vz^i)$.
    \item \textbf{The optimization step}: given $\vx^i = \bar \mLambda \vp^i$, find $\vz^i \in \gZ$ that minimizes $l(\vx^i; \vz)$, identical to the optimization step in the IOA algorithm, which can be solved via off-the-shelf solvers.
\end{itemize}

}
\begin{algorithm}[ht]
   \caption{SO algorithm for solving CNDPs.}
   \label{alg:so}
\begin{algorithmic}[1]
{\footnotesize
   \State {\bfseries Input:} $\vz^0 \in \gZ$, threshold values $\varepsilon > 0$ and $\xi > 0$.
   \vspace{1.5pt}
   \For{$i = 1, 2, \ldots$}
   \vspace{1.5pt}
   \State \textbf{The assignment step}: solve a $\vp^i \in \gP_{\text{so}}^*(\vz^{i - 1})$ via iGP, terminating which when $\delta(\vp^i; \vz^{i - 1}) < \varepsilon$; set $\vx^i = \bar \mLambda \vp^i$.
   \vspace{1.5pt}
    \State \textbf{The optimization step}: solve $\vz^{i =} \in \argmin_{\vz \in \gZ} l(\vx^i; \vz)$.
     \vspace{1.5pt}
    \State If $\| \vz^{i} - \vz^{i - 1}\|_{\infty} < \xi$, break and set $\tilde \vz = \vz^i$.
   \EndFor
}
\end{algorithmic}
\end{algorithm}

{

\noindent
\subsubsection{Heuristics for SCTPs}
\label{app:h-sctp}

For SCTPs, we focus on three heuristics proposed by \citet{harks2015computing}, i.e., MCT, EMCDT, and CT. As mentioned in Section \ref{sec:heuristics}, all three algorithms require first obtaining the system-optimal (SO) link flow pattern $\vx_{\text{so}} = \argmin_{\vx \in \gX}~\langle \vx, u_{\text{time}}(\vx)\rangle.$

\noindent
\textbf{rMCT}.  Algorithm \ref{alg:mct} presents a modified version of MCT, referred to as rMCT in the paper. The only difference between MCT and rMCT lies in Line \ref{ln:decrease}: in MCT, if the condition $\evx_a^i  \leq \evx_{\text{so},a}$ is met, it sets
$
    \evz_a^{i + 1} = \max\{0,  \evz_a^{i} - |\evz_a^{i} - \evz_a^{i - 1}|\},
$
i.e., $\evz_a^{i}$ is decreased by $|\evz_a^{i} - \evz_a^{i - 1}|$, the absolute difference between the objective values in consecutive iterations. The problem is that 
in the first iteration ($i = 1$), this value does not exist since there is no iteration 0.  
In our implementation, instead of reducing $z_a^i$ by the difference, we simply reduce it by a constant that diminishes gradually. We found that rMCT converges smoothly and delivers better solutions than the original version.}

\begin{algorithm}[ht]
   \caption{The revised MCT (rMCT) algorithm for solving SCTPs.}
   \label{alg:mct}
\begin{algorithmic}[1]
{\footnotesize
   \State {\bfseries Input:} $\vz^1 = (\evz_a^1)_{a \in \gA}$ such that $\evz_a^1 = u_{\text{time}, a}'(\evx_{\text{so}, a}) \cdot \evx_{\text{so},a}$ for all $a \in \tilde \gA$ and $\evz_a^1 = 0$ for all $a \in \gA \setminus \tilde \gA$.
   \vspace{1.5pt}
   \For{$i = 1, 2, \ldots$}
   \vspace{1.5pt}
    \State Solve a $\vp^i \in \gP^*(\vz^i)$ via iGP, terminating which when $\delta(\vp^i; \vz^i) < \varepsilon$. Set $\vx^i = \bar \mLambda \vp^i$.
    \State For all $a \in \tilde \gA$, set
    \begin{align*}
        \evz_a^{i + 1} = 
        \begin{cases}
        \evz_a^{i} + 0.9^{i - 1} \cdot u_{\text{time}, a}'(\evx_{a}^i) \cdot \evx_{a}^i, \quad \text{if}~\evx_a^i  > \evx_{\text{so},a} \\
        \max\{0,  \evz_a^{i} - 0.9^{i - 1} \cdot \delta\}, \quad \text{if}~\evx_a^i  \leq \evx_{\text{so},a}.        \end{cases}
    \end{align*}
    \label{ln:decrease}
    \State If $\| \vz^{i} - \vz^{i - 1}\|_{\infty} < \xi$, break and set $\tilde \vz = \vz^i$.
   \EndFor
}
\end{algorithmic}
\end{algorithm}

{
\vspace{0.05in}
\noindent
\textbf{EMCDT and CT.}
The implementation of EMCDT and CT is given  in Algorithms \ref{alg:emcdt} and \ref{alg:ct}, respectively. The reader is referred to \citet{harks2015computing} for the rationale behind these heuristics.

}

\begin{figure}[ht]
\vspace{-0.1in}
\begin{algorithm}[H]
   \caption{The EMCDT algorithm for solving SCTPs.}
   \label{alg:emcdt}
\begin{algorithmic}[1]
{\footnotesize
   \State {\bfseries Input:} $\delta > 0$, $\vz^1 = (\evz_a^1)_{a \in \gA}$ such that $\evz_a^1 = \max\{\delta, u_{\text{time}, a}'(\evx_{\text{so}, a}) \cdot \evx_{\text{so},a}\}$ for all $a \in \tilde \gA$ and $\evz_a^1 = 0$ for all $a \in \gA \setminus \tilde \gA$.
   \vspace{1.5pt}
   \For{$i = 1, 2, \ldots$}
   \vspace{1.5pt}
    \State Solve a $\vp^i \in \gP^*(\vz^i)$ via iGP, terminating which when $\delta(\vp^i; \vz^i) < \varepsilon$. Set $\vx^i = \bar \mLambda \vp^i$.
    \vspace{1.5pt}
   \State Set $\alpha = \max_{a \in \tilde \gA} u_{\text{time}, a}'(\evx_{a}^i) \cdot \evx_{a}^i$. 
   \vspace{1.5pt}
    \State For all $a \in \tilde \gA$, set $\evz_a^{i + 1} = \evz_a^{i} \cdot \exp\big(\frac{0.9^{i - 1}}{\max \{1,  \alpha\}} \cdot (u_{\text{time}, a}'(\evx_{a}^i) \cdot \evx_{a}^i - u_{\text{time}, a}'(\evx_{\text{so}, a}) \cdot \evx_{\text{so},a})\big)$.
    \vspace{1.5pt}
    \State If $\vz^i$ converges, break and set $\tilde \vz = \vz^i$.
   \EndFor
}
\end{algorithmic}
\end{algorithm}
\begin{algorithm}[H]
   \caption{The CT algorithm for solving SCTPs.}
   \label{alg:ct}
\begin{algorithmic}[1]
{\footnotesize
   \State {\bfseries Input:} $\delta > 0$, $\tilde \vz = \vzero$.
   \vspace{1.5pt}
   \While{$\tilde \gA \neq \emptyset$}
   \vspace{1.5pt}
    \State Solve a $\vp^i \in \gP^*(\vz^i)$ via iGP, terminating which when $\delta(\vp^i; \vz^i) < \varepsilon$. Set $\vx^i = \bar \mLambda \vp^i$.
    \vspace{1.5pt}
    \State Find $a' = \argmax_{a \in \tilde \gA} u_{\text{time}, a}'(\evx_{a}^i) \cdot \evx_{a}^i$. If $\evx_{a'} > \evx_{\text{so}, a'}$, set $\evz_{a'} = \evz_{a'} + \delta$; otherwise, remove $a'$ from $\tilde \gA$.
   \EndWhile
}
\end{algorithmic}
\end{algorithm}
\end{figure}

\end{document}